\documentclass[11pt]{article}
\usepackage{graphicx}
\usepackage{amsfonts}
\usepackage{amsmath}
\usepackage{amssymb}
\usepackage{url}
\usepackage{fancyhdr}
\usepackage{indentfirst}
\usepackage{enumerate}
\usepackage{dsfont,footmisc}

\allowdisplaybreaks
\usepackage{multirow}
\newcommand{\M}{\mathcal{M}}
\usepackage{amsthm}
\usepackage{color}
\usepackage{natbib}
\usepackage[british]{babel}
\usepackage{pgfplots}
\pgfplotsset{compat=1.18}
\usetikzlibrary{intersections}

\usepackage[colorlinks=true,citecolor=blue]{hyperref}
\usepackage{tikz}
\usepackage[font=small]{caption} 
\usepackage{soul,xcolor}
\setstcolor{red}
\usepackage{caption}
\usepackage{subcaption}

\usepackage[title]{appendix}

\newcommand{\var}{\mathrm{Var}}

\newcommand{\VaR}{\mathrm{VaR}}

\newcommand{\RVaR}{\mathrm{RVaR}}

\newcommand{\ES}{\mathrm{ES}}
\newcommand{\E}{\mathbb{E}}

\newcommand{\id}{\mathds{1}}

\newcommand{\f}{{\bf f}}
\renewcommand{\a}{{\bf a}}
\renewcommand{\b}{{\bf b}}
\renewcommand{\c}{{\bf c}}

\renewcommand{\ge}{\geqslant}
\renewcommand{\le}{\leqslant}
\renewcommand{\geq}{\geqslant}
\renewcommand{\leq}{\leqslant}
\renewcommand{\epsilon}{\varepsilon}

\theoremstyle{plain}
\newtheorem{theorem}{Theorem}
\newtheorem{corollary}{Corollary}
\newtheorem{lemma}{Lemma}
\newtheorem{proposition}{Proposition}
\theoremstyle{definition}

\theoremstyle{definition}

\newtheorem{remark}{Remark}
\renewcommand{\cite}{\citet}

\usepackage[onehalfspacing]{setspace}

\usepackage{geometry}

\begin{document}

	\title{Pareto-optimal reinsurance under dependence uncertainty}
	
\author{Tim J.~Boonen\thanks{Department of Statistics and Actuarial Science, School of Computing and Data Science, The University of Hong Kong, Hong Kong, China.~Email: \texttt{tjboonen@hku.hk}} \and Xia Han\thanks{School of Mathematical Sciences, LPMC and AAIS, Nankai University, China.  ~Email: \texttt{xiahan@nankai.edu.cn}}  \and Peng Liu\thanks{School of Mathematics, Statistics and Actuarial Science, University of Essex, UK. Email: \texttt{peng.liu@essex.ac.uk}}  \and Jiacong Wang\thanks{School of Mathematical Sciences, Nankai University, China. ~Email: \texttt{2120240092@mail.nankai.edu.cn}}}

	\maketitle 
\begin{abstract}

This paper studies Pareto-optimal reinsurance design in a monopolistic market with multiple primary insurers and a single reinsurer, all with heterogeneous risk preferences. The risk preferences are characterized by  a family of risk measures, called Range Value-at-Risk (RVaR), which includes both Value-at-Risk (VaR) and Expected Shortfall (ES) as special cases. Recognizing the practical difficulty of accurately estimating the dependence structure among the insurers’ losses, we adopt a robust optimization approach that assumes the marginal distributions are known while leaving the dependence structure unspecified. We provide a complete characterization of optimal indemnity schedules under the worst-case scenario, showing that the infinite-dimensional optimization problem can be reduced to a tractable finite-dimensional problem involving only two or three parameters for each indemnity function. Additionally, for independent and identically distributed risks, we exploit the argument of asymptotic normality to derive optimal two-parameter layer contracts. Finally, numerical applications are considered in a two-insurer  setting to illustrate the influence of the dependence structures and heterogeneous risk tolerances on optimal strategies and the corresponding risk evaluation.  
\end{abstract}


\section{Introduction}
Centralized risk pooling, where a single entity assumes the financial risks of a large and heterogeneous client base, underpins modern insurance markets. This principle naturally extends to reinsurance, in which specialized entities absorb and manage risks ceded by primary insurers, enhancing market capacity and stability.

In this paper, we study risk-transfer mechanisms in a monopolistic reinsurance market with multiple primary insurers (cedants) and a single reinsurer. Each cedant holds a fixed portfolio of risks and seeks to transfer part of these risks through structured premium payments. A reinsurance treaty specifies a coordinated schedule of premiums and corresponding indemnity rules. A tension arises between the participants’ perspectives: cedants evaluate the treaty based on their individual post-transfer risk, whereas the reinsurer considers aggregate liabilities versus total premiums.

We evaluate the efficiency of multilateral reinsurance treaties through the lens of Pareto optimality: a treaty is considered efficient if no participant's risk can be reduced without increasing that of another. The seminal work of   \cite{A71} established that for a risk-averse decision maker maximizing expected utility, Pareto-optimal contracts take the form of full coverage above a constant deductible. This foundational result has been substantially generalized to alternative market settings, including those employing distortion-type premium principles and accommodating heterogeneous beliefs; see, e.g., \cite{CLW17}, \cite{JHR18}, \cite{M19}, \cite{BJ23}, and  \cite{CGZ24}. \cite{BCG24} study the case with multiple policyholders transferring risk to a single central authority, achieving Pareto efficiency with distortion risk measures. They provide only an implicit description of the optimal risk-transfer contracts and show the relevance of the setting in examples involving flood risk.

  In our framework, where multiple cedants cede risk to a single reinsurer, presents a distinct challenge, as treaty performance depends not only on the marginal distributions of the cedants' risks but also critically on the dependence structure governing their joint behavior, which directly determines the reinsurer's aggregate loss. \cite{MFE15} emphasize that accurately estimating this dependence structure is notoriously difficult in practice, and its misspecification can lead to severe risk management consequences. Moreover, data for different but correlated insurance lines are often collected separately, providing little or no empirical basis for inferring dependence; see, e.g., \cite{EPR13} and \cite{EWW15}.
  Motivated by these operational challenges, we adopt a robust optimization framework in which the marginal loss distributions are assumed to be known, whereas the dependence structure among the risks is left completely unspecified. 
   Robust optimization provides a principled approach to decision-making under model uncertainty, with its theoretical foundations developed in \cite{BenTal2009}, key methodological advances presented in \cite{BenTal2008} and \cite{Bertsimas2011}, and an overview provided in \cite{Gabrel2014}. Applications of robust methods to insurance and risk management include the analysis of minimax portfolio strategies in \cite{Polak2010} and the study of robust and Pareto-efficient insurance contracts in \cite{Asimit2017}. More recent developments addressing model uncertainty in (re)insurance design have appeared in \cite{Chi2022} and \cite{Cai2024}. Although \cite{FHLX25} also study reinsurance problems under dependence ambiguity, their framework does not incorporate the Pareto-optimal multilateral treaty structure that is central to our analysis.

A defining feature of our model is its ability to coherently accommodate heterogeneous risk preferences among all market participants. We assume that both cedants and the reinsurer evaluate risk using the Range Value-at-Risk (RVaR) measure, while allowing each entity to adopt a distinct RVaR threshold, thereby reflecting differing levels of risk tolerance. The RVaR family, introduced by \cite{CDS10} as a class of robust risk measures, generalizes two of the most widely adopted risk measures in insurance practice and regulation: Value-at-Risk (VaR) and Expected Shortfall (ES). The extensive use of VaR and ES for risk quantification and capital requirements \citep[see, e.g.,][]{CTWZ08,LLM3,CM14} has motivated a substantial body of research on optimal reinsurance design based on these measures. More recently, RVaR has been employed as a preference functional in a range of risk-sharing and reinsurance settings, including cooperative and competitive risk allocation \citep{ELW18} and optimal reinsurance design \citep{GHLLW22,FHLX25}.   Our paper is organized as follows. 

Section~\ref{sec:2} formulates the problem and presents the methodological foundations of our analysis. We show that identifying a Pareto-optimal reinsurance contract is equivalent to solving a system-wide risk minimization problem, where the total risk is represented as a weighted sum of the cedants’ and the reinsurer’s risk exposures (Proposition~\ref{prop:1}).

In Section~\ref{sec:3}, we investigate the worst-case scenario. We assume that the marginal loss distributions of the cedants are known, while the dependence structure among these risks remains completely unspecified. Because robust aggregation results for $\VaR$ and $\RVaR$ are limited, we restrict the search to convex or concave indemnity schedules and work on the corresponding reduced domains. The theoretical developments rely on the techniques of \cite{BLLW20} and their extension to the insurer’s problem in \cite{FHLX25}. Under dependence uncertainty, the problem naturally becomes a minimax optimization in which the objective is to identify indemnity rules that minimize the system’s total risk under the worst possible dependence configuration. Our main analytical contribution in this section is a complete characterization of the Pareto-optimal indemnity schedules. We show that the infinite-dimensional optimization over measurable indemnity functions can be reduced to a finite-dimensional search. Specifically, for each cedant, the optimal indemnity function is either a two- or three-parameter layer contract (Theorem~\ref{th:RVaR}).  These tractable parametric representations not only simplify computation but also provide a clear economic interpretation, illustrating how optimal contracts allocate risk between cedants and the reinsurer under heterogeneous preferences and dependence uncertainty. In the special case where the risk measure is reduced to VaR, we derive explicit forms of the optimal reinsurance contracts (Theorem~\ref{thvar}). For the case of two risks, the optimal contracts admit a particularly simple representation based on Makarov-type bounds, allowing explicit evaluation of the worst-case VaR.

Section~\ref{sec:5} focuses on the case of independent and identically distributed (i.i.d.) risks, which allows us to leverage the asymptotic normality established in Proposition~\ref{CLT} to obtain approximately optimal reinsurance strategies. In this setting, the optimal indemnity function admits a simple two-parameter layer form (Theorem~\ref{RVaR_as}), making the solution more tractable while retaining economic interpretability. In the special case of a single insurer, this formulation coincides with the classical problem of designing an optimal reinsurance contract under a mean-standard deviation framework with $\RVaR$ as the risk criterion, and naturally recovers the $\VaR$- and $\ES$-based results of \cite{C12}.  We also provide the corresponding asymptotic normality results under $\VaR$ (Theorem~\ref{thme}). 

In Section~\ref{sec:6}, we provide numerical illustrations of our theoretical findings. In a two-insurer setting, we derive general $\VaR$-based optimal solutions and compute specific reinsurance contracts for simulated data under three dependence regimes: independence, comonotonicity, and full dependence uncertainty. As expected, worst-case dependence generally produces the largest system risk, although $\VaR$ does not always attain its maximum in the comonotonic case; in some situations, the i.i.d. scenario may yield even larger values than the comonotonic benchmark.  Finally, we present comparative examples illustrating how constraints on reinsurance strategies interact with different loss distributions, highlighting the impact of distributional features on the structure of the optimal indemnity schedules. Section~\ref{sec:7} concludes the paper. The proofs are delegated to the appendices.

\section{Model description and notation}\label{sec:2}
 
 For an atomless probability space  $(\Omega, \mathcal F,\mathbb P)$, let
$L^1$ be the set of random variables with finite expectation, and $L^0$ be the set of all measurable random variables. For simplicity, let  $[n] = \{1,\dots, n\}$. Random variables $X$ and $Y$ are called comonotonic if there exist non-decreasing functions $h$ and $g$ such that $X=h(X+Y)$ and $Y=g(X+Y)$ \citep{D94}.

In a monopolistic reinsurance market with multiple primary insurers and a single reinsurer, we assume that each insurer seeks to purchase an optimal reinsurance contract from the reinsurer. Let $f_i$ denote the indemnity function that maps losses to indemnities.  
To mitigate potential ex post moral hazard, we focus on the class of indemnity functions
\[
\mathcal{I} = \left\{f: [0,\infty) \to [0,\infty) \;\middle|\; f(0)=0, \; 0 \le f(x_2)-f(x_1) \le x_2-x_1 \ \text{for all } 0 \le x_1 \le x_2 \right\}.
\]
This class $\mathcal{I}$ is sufficiently rich: it includes many commonly used indemnity functions, such as the excess-of-loss function $f(x)=(x-d)_+$ for some $d \ge 0$, where $z_+:=\max\{z,0\}$, and the quota-share function $f(x)=q x$ for $q \in [0,1]$.

Given the reinsurance contract  $f_i$ and the premium $\pi_i\in\mathbb{R},$ 
  the loss random variable for the $i$-th insurer  is 
$$
T_{f_i,\pi_i}(X_i)=X_i-f_i(X_i)+\pi_i,$$
and  the loss random variable for the (centralized) reinsurer is 
$$ R_{\mathbf f,\boldsymbol{\pi}}(\mathbf X)=\sum_{i=1}^nf_i(X_i)-\sum_{i=1}^n\pi_i,$$
with  $\mathbf f= (f_1,\dots,f_n),$ $\boldsymbol{\pi}=(\pi_1,\dots,\pi_n)$ and $\mathbf X=(X_1,\dots, X_n).$ In this paper, we assume that the marginal distributions of the risks are fixed, but their dependence structure is left unspecified. The corresponding uncertainty set is defined as
$$
\mathcal{E}_n(\mathbf{F})=\left\{\left(X_1, \ldots, X_n\right): X_i \sim F_i, i\in [n]\right\},
$$
where $\mathbf{F}=\left(F_1, \ldots, F_n\right).$


Next, we introduce the risk measures used in this paper to evaluate the risk.
Define the left quantile of a distribution $F$ and a random variable $X$ with $X\sim F$ at $\alpha\in (0,1]$ as
 $$F^{-1}(\alpha)=\mathrm{VaR}_\alpha(X)=\inf\{x: F(x)\geq\alpha\},$$  
and for $\alpha\in[0,1),$  the right quantile  is given by
 $$F^{-1}_{+}(\alpha)=\mathrm{VaR}^+_\alpha(X)=\inf\{x: F(x)>\alpha\}$$
 with the convention that $\inf\emptyset=\infty$. 
 
In this paper,  we assume that  all insurers and the reinsurer evaluate their risks using Range Value-at-Risk ($\RVaR$) , possibly at different levels. 
 For any $\alpha, \beta$ satisfying  $0\leq \beta<\beta+\alpha\leq 1,$ the $\RVaR$  of a random variable $X \in L^1$ at levels $(\alpha, \beta)$ is defined as

$$
\operatorname{RVaR}_{\beta, \alpha}(X)=\frac 1 \alpha \int_{\beta}^{\alpha+\beta}\VaR_{1-\gamma} (X)\mathrm  d \gamma.   
$$
Note that $\operatorname{RVaR}$ falls in the family of distortion risk measures. Hence, it satisfies the properties enjoyed by the general distortion risk measures such as monotonicity, cash invariance, and  comonotonic additivity; see Chapter 4 of  \cite{FS16}.
Moreover, the two regulatory  risk measures Value-at-Risk $(\VaR)$ and expected shortfall $(\ES)$ are special cases or limits of $\RVaR$. Specifically,
for $\beta\in (0,1)$ and $X\in L^0$, 
$$\mathrm{VaR}_{1-\beta}(X)=\lim_{\alpha\downarrow 0}\RVaR_{\beta,\alpha}(X),~\text{and}~\mathrm{VaR}_{1-\beta}^+(X)=\lim_{\alpha\downarrow 0}\RVaR_{\beta-\alpha,\alpha}(X),$$
and for $\alpha\in [0,1)$ and $X\in L^1$,
$$
\mathrm{ES}_\alpha(X)=\RVaR_{0, 1-\alpha}(X)=\frac{1}{1-\alpha} \int_{\alpha} ^1\VaR_{\gamma} (X)  \mathrm{d} \gamma.
$$







In this paper, our aim is to find the optimal reinsurance policies from the perspective of both the insurers and the reinsurer in the worst-case scenario under dependence uncertainty.  A contract $(\mathbf f,\boldsymbol{\pi}) \in\mathcal I^n\times \mathbb{R}^n$ is said to be \emph{robust Pareto-optimal} under dependence uncertainty if there exists no other contract  $(\hat{\mathbf f},\hat{\boldsymbol{\pi}}) \in\mathcal I^n \times \mathbb{R}^n$   such that 
\begin{align*}
\mathrm{RVaR}_{\beta_i,\alpha_i}\left(T_{f_i,\pi_i}(X_i)\right)&\geq \mathrm{RVaR}_{\beta_i,\alpha_i}\left(T_{\hat f_i,\hat \pi_i}(X_i)\right),\text{ for all } i\in [n],\\
\sup_{\mathbf X\in \mathcal{E}_n(\mathbf{F})}\mathrm{RVaR}_{\beta,\alpha}\left(R_{\mathbf f,\boldsymbol{\pi}}(\mathbf X)\right) &\geq \sup_{\mathbf X\in \mathcal{E}_n(\mathbf{F})}\mathrm{RVaR}_{\beta,\alpha}\left(R_{\hat{\mathbf f},\hat{\boldsymbol{\pi}}}(\mathbf X)\right),
\end{align*}
with at least one of these inequalities being strict. This concept highlights the inherent trade-offs between the insurers' individual risk assessments and the reinsurer's evaluation under the worst-case dependence structure.

The following proposition shows that robust Pareto-optimal contracts can be characterized through a single aggregated optimization problem.
\begin{proposition}\label{prop:1}
A contract $(\mathbf f,\boldsymbol{\pi}) \in\mathcal I^n\times \mathbb{R}^n$  is robust Pareto-optimal under dependence uncertainty if and only if  $\mathbf f \in\mathcal I^n$ solves $\inf_{\mathbf f \in\mathcal I^n}V(\mathbf f),$ where 
 \begin{equation}\label{eq:prob1}
V(\mathbf f): =\sum_{i=1}^n \mathrm{RVaR}_{\beta_i,\alpha_i}\left(T_{f_i,\pi_i}(X_i)\right)+ \sup_{\mathbf X\in \mathcal{E}_n(\mathbf{F})}\mathrm{RVaR}_{\beta,\alpha}\left(R_{\mathbf f,\boldsymbol{\pi}}(\mathbf X)\right),~~~\bf f\in\mathcal I^n.
\end{equation}
\end{proposition}


Note that $V(\f)$ is independent of $\boldsymbol{\pi}$ and hence we rewrite it as 
\begin{equation}\label{eq:V2}V(\mathbf f)=\sum_{i=1}^n \mathrm{RVaR}_{\beta_i,\alpha_i}\left(T_{f_i}(X_i)\right)+ \sup_{\mathbf X\in \mathcal{E}_n(\mathbf{F})}\mathrm{RVaR}_{\beta,\alpha}\left(R_{\mathbf f}(\mathbf X)\right),\end{equation}
where 
$T_{f_i}(X_i)=X_i-f_i(X_i)$ and 
$R_{\mathbf f}(\mathbf X)=\sum_{i=1}^nf_i(X_i).$
In the following, we will focus on the optimal reinsurance policies such that   $V(\f)$  is minimized over $\f\in\mathcal I^n.$
\begin{remark}
In our framework, the premium $\pi_i$ for each insurer is treated as a fixed constant, agreed upon in advance between the insurer and the reinsurer. This reflects practical situations in which framework contracts or long-term agreements specify the premium upfront, leaving only the indemnity structure $f_i$ adjustable by the insurer to manage retained risk. Under this assumption, the reinsurer's net loss $R_{\mathbf f,\boldsymbol{\pi}}$ depends on the indemnity functions up to an additive constant determined by the premiums.

A natural concern arises: if the reinsurer strictly seeks to minimize its risk, it could in principle choose $f_i = 0$, effectively avoiding any risk exposure. While this would indeed minimize the reinsurer's risk, in practice, the reinsurer may still accept a positive $f_i$ for several reasons: contractual obligations, regulatory requirements, long-term relationship considerations, or specified risk tolerances. Within this setup, Pareto optimality  is  meaningful: it identifies indemnity schedules where no participant—whether an insurer or the reinsurer—can reduce their own risk without increasing the risk of another, given the fixed premiums and the risk tolerances of all parties. Thus,  with fixed premiums, Pareto-optimal contracts capture the efficient sharing of risk across insurers and the reinsurer, provided the reinsurer's willingness to bear risk is bounded by these practical considerations.

From a theoretical perspective, treating the premium as fixed allows us to focus on the structure of Pareto-optimal reinsurance contracts under heterogeneous risk preferences and dependence uncertainty, without the additional complexity of a premium that varies with $f_i$. 

\end{remark}

\section{Optimal insurance with dependence uncertainty}\label{sec:3}
Note that Proposition \ref{prop:1} implies that identifying a robust Pareto-optimal contract is equivalent to minimizing \(V(\mathbf f)\); that is, determining the optimal reinsurance strategies for individual insurers so that the total risk of a system with multiple insurers and a single reinsurer is minimized in the worst-case scenario under dependence uncertainty,
i.e., finding the optimal reinsurance policies $\mathbf{f} \in \mathcal{I}^n$ that solve:
\begin{equation}\label{eq:opt1}
\inf_{\mathbf f \in\mathcal I^n}\left\{\sum_{i=1}^n \mathrm{RVaR}_{\beta_i,\alpha_i}\left(T_{f_i}(X_i)\right)+ \sup _{\mathbf X \in \mathcal{E}_n(\mathbf{F})} \mathrm{RVaR}_{\beta,\alpha}\left(R_{\mathbf f}(\mathbf X)\right)\right\}.
\end{equation}


Due to the limited availability of robust aggregation results for VaR and RVaR, it is sometimes necessary to impose additional restrictions on the indemnity functions, such as convexity or concavity. Accordingly, we consider the following domains:
\[
\mathcal{I}_{cx}^n = \left\{ \mathbf f = (f_1,\ldots,f_n) \in \mathcal I^n :\ f_i \text{ is convex for } i\in [n] \right\},
\]
and
\[
\mathcal{I}_{cv}^n = \left\{ \mathbf f = (f_1,\ldots,f_n) \in \mathcal I^n :\ f_i \text{ is concave for } i\in [n] \right\}.
\]

For $\alpha\in (0,1),$ we say that {a distribution $F$ is concave beyond its $\alpha$-quantile if the distribution ${(F(x)-\alpha)_+}/{(1-\alpha)}$ is concave over $(F_+^{-1}(\alpha),\infty),$ and  a distribution $F$ is convex beyond its $\alpha$-quantile if the distribution ${(F(x)-\alpha)_+}/{(1-\alpha)}$ is convex on $(-\infty, F^{-1}(1))$; 
see, for instance, Figure \ref{Figure1}. We denote by $\mathcal{M}_{c x}^\alpha$ the set of all distributions that are convex beyond their corresponding $\alpha$-quantiles, and by $\mathcal{M}_{c v}^\alpha$ the set of all distributions that are concave beyond their corresponding $\alpha$-quantiles.  

\vspace{1cm}

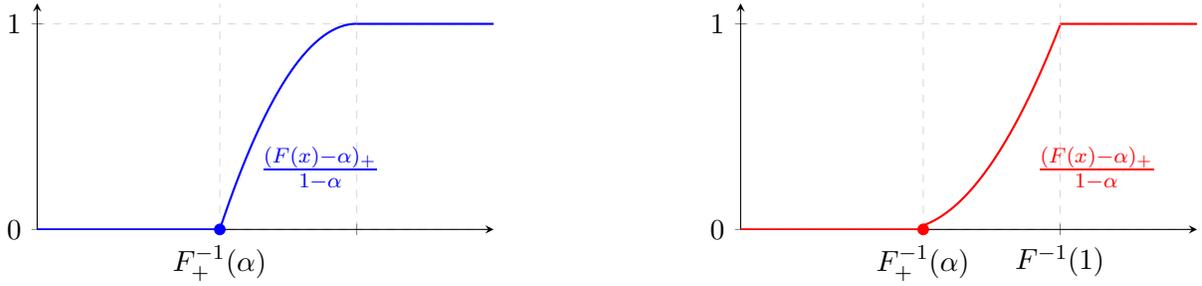
\begin{figure}[htbp]
\centering
\begin{subfigure}{0.45\textwidth}
\centering
\begin{tikzpicture}
\begin{axis}[
    width=\textwidth,
    height=0.6\textwidth,
    axis lines = left,
    xlabel = {},
    ylabel = {},
    ymin=0, ymax=1.1,
    xmin=0, xmax=5,
    xtick = {2,3.5},
    xticklabels = {$F^{-1}_+(\alpha)$},
    ytick = {0,1},
    yticklabels = {0,1},
    grid = major,
    grid style = {dashed, gray!30},
    samples=100,
]

\addplot[domain=0:2, blue, thick] {0};
\addplot[domain=2:3.5, blue, thick] {-0.444*(x-2)^2 + 1.333*(x-2) - 0.0001};
\addplot[domain=3.5:5, blue, thick] {1};
\addplot[only marks, mark=*, blue] coordinates {(2,0)};
\node[blue] at (axis cs: 3.1,0.3) {$\frac{(F(x)-\alpha)_+}{1-\alpha}$};
\end{axis}
\end{tikzpicture}
\end{subfigure}
\hfill
\begin{subfigure}{0.45\textwidth}
\centering
\begin{tikzpicture}
\begin{axis}[
    width=\textwidth,
    height=0.6\textwidth,
    axis lines = left,
    xlabel = {},
    ylabel = {},
    ymin=0, ymax=1.1,
    xmin=0, xmax=5,
    xtick = {2,3.5},
    xticklabels = {$F^{-1}_+(\alpha)$,$F^{-1}(1)$},
    ytick = {0,1},
    yticklabels = {0,1},
    grid = major,
    grid style = {dashed, gray!30},
    samples=100,
]

\addplot[domain=0:2, red, thick] {0};
\addplot[domain=2:3.5, red, thick] {0.3265*(x-1.755)^2};
\addplot[domain=3.5:5, red, thick] {1};
\addplot[only marks, mark=*, red] coordinates {(2,0)};
\node[red] at (axis cs: 3.9,0.3) {$\frac{(F(x)-\alpha)_+}{1-\alpha}$};
\end{axis}
\end{tikzpicture}

\end{subfigure}
\caption{Concave beyond the $\alpha$-quantile (left panel) and convex beyond the $\alpha$-quantile (right panel).}\label{Figure1}
\end{figure}

The following representation for robust $\RVaR$ with dependence uncertainty comes from Proposition 4 of \cite{FHLX25}, which is crucial for the proof of our Theorem \ref{th:RVaR}.
Define \begin{equation}\label{eq:delta}\Delta_n=\left\{\pmb{\gamma}\in (0,1)\times [0,1)^n: \sum_{i=0}^{n}\gamma_i=1\right\},~\text{with}~\pmb{\gamma}=(\gamma_0,\gamma_1,\dots,\gamma_n).\end{equation}
Furthermore,  for any $\alpha\in (0,1),$ we define $\alpha\Delta_n=\left\{\pmb{\gamma}\in (0,1)\times [0,1)^n: \sum_{i=0}^{n}\gamma_i=\alpha\right\}.$
\begin{lemma}\label{lem:0} Suppose that $F_1^{-1},\dots, F_n^{-1}$ are continuous on $(0,1).$     For any $\alpha, \beta$ satisfying $0\leq \beta<\beta+\alpha\leq 1,$  if 
$\mathbf F\in (\M_{cv}^{1-\beta-\alpha})^n,$ 
then we have 
\begin{align}\label{rvar1}\sup_{\mathbf X \in \mathcal{E}_n(\mathbf F)}\mathrm{RVaR}_{\beta,\alpha}\left(\sum_{i=1}^{n}X_i\right)=\inf_{\pmb{\gamma}\in(\beta+\alpha)\Delta_n,\gamma_0\geq \alpha}\sum_{i=1}^{n}\RVaR_{\gamma_i,\gamma_0}(X_i).
\end{align}
\end{lemma}

The proof of Lemma~\ref{lem:0} follows the arguments in \cite{BLLW20}, which establish the representation for distributions with decreasing densities in the tail part. The extension to distributions that are concave in the tail region is provided in \cite{FHLX25}, thereby establishing Lemma~\ref{lem:0}.

Define
 $g_{a,b}(x):=(x-a)_+-(x-b)_+,$
 where $0\leq a\leq  b\leq \infty$. This function represents a layered coverage, paying losses exceeding a retention level
$a$ that are capped at 
$b$.   The special case  $b=\infty$ corresponds to the classical stop-loss indemnity, which covers all losses above the retention level.  We denote
 ${\bf g}_{\a,\b}=(g_{a_1,b_1},\dots,g_{a_n,b_n})$ with parameter domain
 \begin{equation}\label{eq:a1}\begin{aligned}
     \mathcal{A}_1=\{({\bf a}, {\bf b}): {\bf g}_{\a,\b}\in \mathcal I^n, 0\leq a_i\leq b_i\leq \infty,~i\in [n]\}.
 \end{aligned}
 \end{equation}

 Define $r_{a,b,c}(x):=a x+c(x-b)_+ $ with $0\leq a,c\leq a+c\leq 1$ and $0\leq b\leq \infty.$ This class allows a flexible combination of proportional and excess-of-loss strategies,  providing fine-grained control over both the proportion of loss retained and the additional protection above specified thresholds. The special case 
$b=\infty$  is reduced to a pure quota share, in which a fixed proportion of all losses is ceded. 
 We denote
 ${\bf r}_{\a,\b, \c}=(r_{a_1,b_1,c_1},\dots,r_{a_n,b_n,c_n})$ with parameter domain
 $$\mathcal{A}_2=\{({\bf a}, {\bf b},\c): {\bf r}_{\a,\b,\c}\in\mathcal I^n,  0\leq a_i,c_i\leq a_i+c_i\leq 1, b_i\geq 0, ~i\in [n]\}.$$  
In what follows, we use the convention that $\frac{0}{0}=0.$

Our main result is stated as follows. 

\begin{theorem}\label{th:RVaR} Let $V(\mathbf f)$ be given by \eqref{eq:V2}, and $\Delta_n$ be given by \eqref{eq:delta}.  Suppose $F_1^{-1},\dots, F_n^{-1}$ are continuous on $(0,1).$ For any $\alpha, \beta$ satisfying $0\leq \beta<\beta+\alpha\leq 1,$ we have the following two conclusions.
\begin{enumerate}
\item[(i)] If $\beta=0,$  then 
$$\begin{aligned}\inf_{\mathbf f \in\mathcal I^n}  V(\f)=  \inf_{(\a,\b)\in\mathcal{A}_1}   G(\a,\b),\end{aligned}$$ where \begin{equation}\label{eq:G}G({\bf a}, {\bf b})=\sum_{i=1}^n \left\{\mathrm{RVaR}_{\beta_i,\alpha_i}(X_i) -  \mathrm{RVaR}_{\beta_i,\alpha_i}\left(g_{a_i,b_i}(X_i)\right)+  \mathrm{ES}_{1-\alpha}(g_{a_i,b_i}(X_i))\right\}.\end{equation}
\item[(ii)] If $\mathbf F\in (\M_{cv}^{1-\beta-\alpha})^n,$ then 
$$\begin{aligned}\inf_{\mathbf f \in\mathcal I_{cx}^n}  V(\f)=  \inf_{(\a,\b,\c)\in\mathcal{A}_2} \inf_{\pmb{\gamma}\in(\beta+\alpha)\Delta_n,\gamma_0\geq \alpha} R(\a,\b,\c,\pmb\gamma),\end{aligned}$$ where \begin{equation}\label{eq:R}R({\bf a}, {\bf b},\c,\pmb\gamma)=\sum_{i=1}^n \left\{\mathrm{RVaR}_{\beta_i,\alpha_i}(X_i) -  \mathrm{RVaR}_{\beta_i,\alpha_i}\left(r_{a_i,b_i,c_i}(X_i)\right)+  \mathrm{RVaR}_{\gamma_i,\gamma_0}(r_{a_i,b_i,c_i}(X_i))\right\}.\end{equation}
\end{enumerate}
\end{theorem}

The forms of optimal indemnity functions are given in Theorem \ref{th:RVaR}. To determine the corresponding optimal reinsurance strategies, it remains to specify the parameters of these functions. This is addressed in the following proposition.

\begin{proposition}\label{exist} Let $G({\bf a}, {\bf b})$ and $R({\bf a}, {\bf b},\c, \pmb\gamma)$ be given by \eqref{eq:G} and \eqref{eq:R}, respectively. The parameters of the optimal ceded loss functions ${\bf g}_{\a,\b}$ and ${\bf r}_{\a,\b,\c}$, as defined in (i)–(ii) of Theorem \ref{th:RVaR}, are well defined and are determined as follows
 \begin{enumerate}
   \item[(i)] For ${\bf g}_{\a,\b}$,
  $$({\bf a}^*, {\bf b}^*)\in \arg\inf_{(\a,\b)\in\mathcal{A}_1}
G({\bf a}, {\bf b});$$
 \item[(ii)] For ${\bf r}_{\a,\b,\c}$,
 $$(\a^*,\b^*,\c^*)=\arg\inf_{(\a,\b,\c)\in\mathcal{A}_2}
  \left\{\inf_{\pmb{\gamma}\in(\beta+\alpha)\Delta_n,\gamma_0\geq \alpha}R({\bf a}, {\bf b},\c, \pmb\gamma)\right\}.$$
\end{enumerate} 
\end{proposition}

Theorems \ref{th:RVaR} and Proposition \ref{exist} provide a complete characterization of optimal reinsurance designs under the robust framework with multiple insurers and dependence uncertainty. Our analysis establishes that layered stop-loss indemnities $g_{\a,\b}$ are optimal when evaluating risk using ES, while combined proportional-excess-of-loss contracts $r_{\a,\b,\c}$ emerge as optimal under RVaR criteria for convex indemnities. This characterization effectively reduces the inherently infinite-dimensional optimization over admissible indemnity functions to tractable finite-dimensional parameter search problems over the domains $\mathcal{A}_1$ and $\mathcal{A}_2$.
Proposition \ref{exist} further ensures the existence of optimal parameters and guarantees that these solutions can be obtained numerically. This transforms complex reinsurance contract design problems into computationally manageable optimization tasks, bridging theoretical optimality with practical implementability.

Note that the optimal indemnity functions derived in Theorem 2 of \cite{FHLX25} have the form of $a\min(x,b)$ with $0\leq a\leq 1$ and $0\leq b\leq \infty$, which is very different from the optimal indemnity  functions obtained in Theorem \ref{th:RVaR} with the form $r_{a,b,c}(x)=a x+c(x-b)_+ $ with $0\leq a,c\leq a+c\leq 1$ and $0\leq b\leq \infty.$ This is due to the qualitatively different setups of the models in the two papers. The objective in \cite{FHLX25} is to minimize the total risk from the perspective of the reinsurer under dependence uncertainty, whereas the objective in our paper is to find the robust Pareto-optimal contract for a system consisting of multiple insurers and a single reinsurer. Moreover, the proof of Theorem \ref{th:RVaR} is more complex as it involves the detailed discussion of six different cases.

\section{Optimal  solution to the special case of   VaR}\label{sec:4}

In this section, we examine the VaR-based optimal insurance problem, which represents a special case of RVaR. Recall that the relationship between these measures is given by $\mathrm{VaR}_{1-\beta}(X) = \lim_{\alpha \downarrow 0} \RVaR_{\beta,\alpha}(X)$.   We next aim to minimize the following target:
\begin{equation}\label{pro:VaR}
    V(\mathbf f)=\sum_{i=1}^n \VaR_{\alpha_i}\left(T_{f_i}(X_i)\right)+ \sup_{\mathbf X\in \mathcal E_n(\mathbf F)}\VaR_{\alpha}\left(R_{\mathbf f}(\mathbf X)\right). \end{equation}


 An analogous representation for $\RVaR$ in  Lemma~\ref{lem:0} can be established for $\mathrm{VaR}$, which can also be extended to $\mathbf F \in (\mathcal M_{cx}^\alpha)^n$ in this special case.

\begin{lemma}\label{lem:1}
 Suppose $F_1^{-1},\dots, F_n^{-1}$ are continuous on $(0,1)$ and $\mathbf{F} \in\left(\mathcal{M}_{c x}^\alpha\right)^n \cup\left(\mathcal{M}_{c v}^\alpha\right)^n,$ then we have

$$
\sup _{\mathbf X \in \mathcal{E}_n(\mathbf{F})} \operatorname{VaR}_{\alpha}\left(\sum_{i=1}^n X_i\right)=\inf _{\gamma \in(1-\alpha) \Delta_n} \sum_{i=1}^n \RVaR_{\gamma_i, \gamma_0}\left(X_i\right). 
$$

\end{lemma}

\begin{proof}
By Lemma 4.5 of \cite{BJW14}, the continuity of $F_1^{-1},\dots, F_n^{-1}$ over $(0,1)$ implies that for $\alpha\in (0,1),$
$$\sup_{\mathbf X \in \mathcal{E}_n(\mathbf F)}\mathrm{VaR}_{\alpha}^+\left(\sum_{i=1}^nX_i\right)
= \sup_{\mathbf X \in \mathcal{E}_n(\mathbf F)}\mathrm{VaR}_{\alpha}\left(\sum_{i=1}^nX_i\right).$$
Applying Proposition 1 of \cite{FHLX25}, we obtain the desired result.
\end{proof} 
We introduce two additional reinsurance policies that will be used in our main results. Define
 $l_{a,b}(x):=a\min(x,b),$
 where $0\leq a\leq 1$ and $0\leq b\leq \infty,$ including the quota-share function as a special case (when $b=\infty$).  Let
 ${\bf l}_{\a,\b}=(l_{a_1,b_1},\dots,l_{a_n,b_n})$ with parameter domain
 $$\mathcal{A}_3=\{({\bf a}, {\bf b}): {\bf l}_{\a,\b}\in \mathcal I^n, 0\leq a_i\leq 1, 0\leq b_i\leq \infty,~i=1,\dots,n\}.$$ 
 Moreover, define
 $h_{a,b}(x):=a(x-b)_+,$ with  $ 0\leq a\leq\infty$ and $0\leq b\leq 1,$ including the quota-share ($b=0$) and stop-loss ($a=1$) functions as  special cases.  Let ${\bf h}_{\a,\b}=(h_{a_1,b_1},\dots,h_{a_n,b_n})$ with parameter domain $$\mathcal{A}_4=\{({\bf a}, {\bf b}): {\bf h}_{\a,\b}\in \mathcal{I}^n,  0\leq a_i\leq 1, 0\leq b_i\leq\infty, ~i=1,\dots,n\}.$$  
Note that $\mathcal{A}_4$ constitutes a subset of $\mathcal{A}_2$, indicating that when the RVaR risk measure degenerates to VaR, the structure of optimal ceded loss functions becomes more specific and constrained.

\begin{theorem}\label{thvar}  Let $V(\mathbf f)$ be given by \eqref{pro:VaR}. Suppose $F_1^{-1},\dots, F_n^{-1}$ are continuous on $(0,1)$ and $\alpha\in (0,1).$ We have the following conclusions.
\begin{enumerate}
\item[(i)] If $n=2$  then
\begin{align*}
\inf_{\f\in\mathcal{I}^n}V(\mathbf f) =\inf_{(\a,\b)\in\mathcal A_1}
  \inf_{\pmb{\gamma}\in(1-\alpha)\Delta_n} \overline G(\a, \b, {\pmb \gamma}),
\end{align*}
where 
\begin{equation}\label{bgab}\begin{aligned}
    \overline G({\bf a}, {\bf b}, {\pmb \gamma})=\sum_{i=1}^n \left\{\mathrm{VaR}_{\alpha_i}(X_i)-\VaR_{\alpha_i}(g_{a_i,b_i}(X_i))+\RVaR_{\gamma_i,\gamma_0}(g_{a_i,b_i}(X_i))\right\}.
\end{aligned}
\end{equation}
\item[(ii)] If $\mathbf F\in (\M_{cx}^{\alpha})^n,$ then
\begin{align*}
  \inf_{\f\in\mathcal{I}_{cv}^n}V(\mathbf f)  =\inf_{(\a,\b)\in\mathcal{A}_3}
  \inf_{\pmb{\gamma}\in(1-\alpha)\Delta_n} L(\a, \b, {\pmb \gamma}),
\end{align*}
where 
$$\begin{aligned}
    L(\a, \b, {\pmb \gamma})=\sum_{i=1}^n \left\{\mathrm{VaR}_{\alpha_i}(X_i)-\VaR_{\alpha_i}(l_{a_i,b_i}(X_i))+\RVaR_{\gamma_i,\gamma_0}(l_{a_i,b_i}(X_i))\right\}.
\end{aligned}$$
\item[(iii)] If $\mathbf F\in (\M_{cv}^{\alpha})^n,$ then
\begin{align*}
  \inf_{\f\in\mathcal{I}_{cx}^n}V(\mathbf f) =\inf_{(\a,\b)\in\mathcal{A}_4}
  \inf_{\pmb{\gamma}\in(1-\alpha)\Delta_n} H(\a, \b, {\pmb \gamma}),
\end{align*}
where 
$$\begin{aligned}
    H(\a, \b, {\pmb \gamma})=\sum_{i=1}^n \left\{\mathrm{VaR}_{\alpha_i}(X_i)-\VaR_{\alpha_i}(h_{a_i,b_i}(X_i))+\RVaR_{\gamma_i,\gamma_0}(h_{a_i,b_i}(X_i))\right\}.
\end{aligned}$$
\end{enumerate}
\end{theorem}

Further, for the case $n=2$, we can offer a simpler expression than that of (i) of Theorem \ref{thvar} by applying the result in \cite{M81}, which shows that 
\begin{equation}\label{eq:var2}
    \begin{aligned}
        \sup _{\left(X_1,X_2\right) \in \mathcal{E}_2(\mathbf{F})} \operatorname{VaR}^+_{\alpha}\left( X_1+X_2\right)=\inf_{t\in [0,1-\alpha]} \{\VaR_{\alpha+t}(X_1)+\VaR_{1-t}(X_2) \}.  
    \end{aligned}
\end{equation}

 \begin{proposition}\label{n2} Let $V(\mathbf f)$ be given by \eqref{pro:VaR}.  For $n=2$, suppose that $F_1^{-1}$ and $F_2^{-1}$ are continuous on $(0,1)$, then
\begin{equation*}\label{simple} \begin{aligned}
 \inf_{(f_1,f_2)\in \mathcal I^2} V({\bf f})
=\inf_{(a_1,a_2,b_1,b_2)\in\mathcal A_1}\inf_{t\in[0,1-\alpha]}\overline G_1(a_1,a_2,b_1,b_2,t),
 \end{aligned}\end{equation*}
 where
\begin{equation}\label{eq:G1bar}\begin{aligned}\overline G_1(a_1,a_2,b_1,b_2,t)=&\VaR_{\alpha_1}(X_1)+\VaR_{\alpha_2}(X_2)+ \VaR_{\alpha+t}(g_{a_1,b_1}(X_1))- \VaR_{\alpha_1}(g_{a_1,b_1}(X_1))\\&+\VaR_{1-t}(g_{a_2,b_2}(X_2))- \VaR_{\alpha_2}(g_{a_2,b_2}(X_2)).
 \end{aligned}\end{equation}
Moreover, $(g_{a_1,b_1},g_{a_2,b_2})$ are the optimal indemnity functions for the worst-case scenario, provided
 $$(a_1,a_2,b_1,b_2)\in \arg \inf_{(a_1,a_2,b_1,b_2)\in\mathcal A_1}\left\{\inf_{t\in[0,1-\alpha]}\overline G_1(a_1,a_2,b_1,b_2,t)\right\}.$$
 \end{proposition}

\begin{remark}\label{rem:3}
We observe from \eqref{eq:G1bar} that if either $\alpha \ge \alpha_1$ or $\alpha\ge \alpha_2$, then the following inequalities hold for $t\in [0,1-\alpha]$:
\[
\VaR_{\alpha+t}\big(g_{a_1,b_1}(X_1)\big) - g_{a_1,b_1}\big(\VaR_{\alpha_1}(X_1)\big) \ge 0
\quad \text{or} \quad
\VaR_{1-t}\big(g_{a_2,b_2}(X_2)\big) - g_{a_2,b_2}\big(\VaR_{\alpha_2}(X_2)\big) \ge 0.
\]
Consequently,  the optimal value of the objective function coincides with the no-insurance benchmark, i.e.,
$
\VaR_{\alpha_i}(X_i),~i=1,2.
$
This implies that under these parameter conditions, purchasing insurance does not provide any improvement over the no-reinsurance case using $\VaR$ to quantify the risk.
\end{remark}

To obtain the optimal reinsurance strategies, we need to fix the parameters of these functions in Theorem \ref{thvar}, which will be discussed in the following proposition.
The proof follows along the same lines as that of Proposition \ref{exist} and is therefore omitted. \begin{proposition} 
 The parameters of the optimal ceded loss functions $ {\bf g}_{\a,\b}, {\bf l}_{\a,\b},{\bf h}_{\a,\b}$ in (i)-(iii) of Theorem \ref{thvar} are well defined and are determined as follows. 
   \begin{enumerate}
   \item[(i)] For ${\bf g}_{\a,\b}$,
  $$({\bf a}^*, {\bf b}^*)\in \arg\inf_{(\a,\b)\in\mathcal{A}_1}\left\{
  \inf_{\pmb{\gamma}\in(1-\alpha)\Delta_n}\overline G({\bf a}, {\bf b}, {\pmb \gamma})\right\};$$
 \item[(ii)] For ${\bf l}_{\a,\b}$,
$$ (\a^*, \b^*)=\arg\inf_{(\a,\b)\in\mathcal{A}_3}
  \left\{\inf_{\pmb{\gamma}\in(1-\alpha)\Delta_n} L(\a, \b, {\pmb \gamma})\right\};$$
 \item[(iii)] For ${\bf h}_{\a,\b}$,
 $$(\a^*,\b^*)=\arg\inf_{(\a,\b)\in\mathcal{A}_4}
  \left\{\inf_{\pmb{\gamma}\in(1-\alpha)\Delta_n} H(\a, \b, {\pmb \gamma})\right\}.$$
\end{enumerate}

\end{proposition}

In Theorem 1 of \cite{FHLX25}, the obtained optimal indemnity functions have the form (i) $g_{a,b}$ for $n=2$; (ii) $c(x-a)_+ + d(x-b)_+$ with $0\leq a\leq b\leq \infty$ and $0\leq c,d\leq c+d\leq 1$ for $\mathbf F\in (\M_{cx}^{\alpha})^n$; (iii) $l_{a,b}(x)=a\min(x,b)$ with $0\leq a\leq 1$ and $0\leq b\leq \infty$ for $\mathbf F\in (\M_{cv}^{\alpha})^n$. For the cases of $\mathbf F\in (\M_{cx}^{\alpha})^n$ and $\mathbf F\in (\M_{cv}^{\alpha})^n$, the optimal indemnity functions in \cite{FHLX25} are completely different from those derived in Theorem \ref{thvar} due to the different setups of the two models.   

\section{Optimal solutions for i.i.d. risks}\label{sec:5}
In the previous section, we analyzed the optimal reinsurance problem under fully unknown dependence, where the joint distribution of $(X_1,\dots,X_n)$ was unspecified and a worst-case formulation was necessary. In this section, we focus on the case where the risks $X_1,\dots,X_n$ are i.i.d.. In particular, when $n\to\infty$, the aggregated risk exhibits asymptotic normality, which allows us to derive explicit and tractable expressions for the optimal indemnity functions.  Our target is to minimize
 \begin{equation}\label{eq:prob11}
V(\mathbf f): =\sum_{i=1}^n \mathrm{RVaR}_{\beta_i,\alpha_i}\left(T_{f_i}(X_i)\right)+ \mathrm{RVaR}_{\beta,\alpha}\left(R_{\mathbf f}(\mathbf X)\right),~~~\bf f\in\mathcal I^n.
\end{equation}
The following proposition provides a classical asymptotic normality result for sums of i.i.d.\ transformed risks.

\begin{proposition}\label{CLT}
Let $\{X_i\}_{i=1}^n$ be a sequence of nonnegative i.i.d. random variables with mean $\mathbb{E}[X_i] = \mu<\infty$ and variance $\mathrm{Var}(X_i) = \sigma^2<\infty$. Define $S_n = \sum_{i=1}^n f_i(X_i)$ with $f_i\in \mathcal{I}$. If $\var(S_n)\to\infty$ as $n\to\infty$, then
$$
\frac{S_n - \mathbb{E}[S_n]}{\sqrt{\mathrm{Var}(S_n)}} \xrightarrow{d} \mathcal{N}(0, 1). 
$$
\end{proposition}

The condition $\mathrm{Var}(S_n) \to \infty$ in Proposition~\ref{CLT} captures the natural growth of total accumulated risk as the number of risks increases, assuming that the $f_i$ are nontrivial (not identically zero).

Because of the conclusion in Proposition \ref{CLT},  we suppose that $R_{\mathbf f}(\mathbf X)\sim \mathcal{N}(\E(R_{\mathbf f}(\mathbf X)), \mathrm{Var}(R_{\mathbf f}(\mathbf X)))$ in the following two subsections.

\subsection{The results for RVaR} Recall that 
$
\bar \alpha_i= 1-\beta_i-\alpha_i,  \bar \beta_i =1-\beta_i, 
\bar \alpha= 1-\beta-\alpha,$ and $\bar \beta =1-\beta.
$
Then the problem in equation \eqref{eq:prob11} asymptotically becomes
\begin{equation}\label{pro:1}
\min_{\mathbf f\in \mathcal I^n} \widetilde V(\mathbf f) 
:= \sum_{i=1}^n \mathrm{RVaR}_{\beta_i, \alpha_i} (X_i - f_i(X_i)) 
   + \mu(\mathbf f) 
   + \sigma(\mathbf f) \frac{1}{\alpha} \int_{\bar\alpha}^{\bar\beta} \Phi^{-1}(\gamma) \mathrm d\gamma,
\end{equation}
where $\Phi$ is the cumulative distribution function of a standard normal random variable, and
\[
\mu(\mathbf f) := \sum_{i=1}^n \mathbb{E}[f_i(X_i)], \qquad 
\sigma^2(\mathbf f) := \sum_{i=1}^n \mathrm{Var}(f_i(X_i)).
\]
\begin{remark}
The optimization problem in \eqref{pro:1} admits an alternative interpretation in the context of collective reinsurance purchasing. Specifically, it can be viewed as minimizing the aggregate risk exposure of $n$ insurers who jointly purchase reinsurance for their respective business lines, with premiums calculated via a mean-standard deviation principle, where $\frac{1}{\alpha} \int_{\bar\alpha}^{\bar\beta} \Phi^{-1}(\gamma) \mathrm d\gamma$ serves as a loading coefficient. 
In the special case when $n=1$, our framework reduces to the single-insurer problem and generalizes part of the results of \citet{C12}, who considered optimal reinsurance under VaR and ES criteria with mean-standard deviation premium principles. \end{remark}

The following lemma is well known; see, e.g., Property 3.4.19 in \cite{DDGK05} and Lemma A.2 in \cite{C12}.
\begin{lemma} \label{lem:3}Provided that the random variables $Y_1$ and $Y_2$ have finite expectations, if they satisfy
$$
\mathbb{E}\left[Y_1\right]=\mathbb{E}\left[Y_2\right], \quad F_{Y_1}(t) \leq F_{Y_2}(t),\quad t<t_0,~~S_{Y_1}(t) \leq S_{Y_2}(t),\quad t \geq t_0
$$
for some $t_0 \in \mathbb{R}$, then $Y_1 \leq_{c x} Y_2$, i.e.,
$$
\mathbb{E}\left[G\left(Y_1\right)\right] \leq \mathbb{E}\left[G\left(Y_2\right)\right]
$$
for any convex function $G(x)$ provided that the expectations exist.
\end{lemma}
Theorem~\ref{RVaR_as} characterizes the explicit {form of the optimal retention strategies} in the asymptotic framework. 
\begin{theorem} \label{RVaR_as}
Let $ \widetilde  V(\f)$ be given by \eqref{pro:1}. We have 
$$\begin{aligned}\inf_{\mathbf f \in\mathcal I^n} \widetilde  V(\f)=  \inf_{(\a,\b)\in\mathcal{A}_1}   \widetilde  G(\a,\b),\end{aligned}$$ where \begin{equation}\label{eq:tildeG}\begin{aligned}\widetilde G({\bf a}, {\bf b})&=\sum_{i=1}^n \left\{\mathrm{RVaR}_{\beta_i,\alpha_i}(X_i) -  \mathrm{RVaR}_{\beta_i,\alpha_i}\left(g_{a_i,b_i}(X_i)\right)+  \mathbb{E}[g_{a_i,b_i}(X_i)] \right\}
\\&+\left(\sum_{i=1}^n \mathrm{Var}(g_{a_i,b_i}(X_i))\right)^{1/2} \frac{1}{\alpha}\int_{\bar\alpha}^{\bar\beta} \Phi^{-1}(\gamma)\mathrm d\gamma.\end{aligned}\end{equation}
\end{theorem}

The layered structure of the optimal strategy in Theorem \ref{RVaR_as} under the asymptotic normality framework aligns with the forms identified in worst-case scenarios where RVaR degenerates to ES (when $\beta = 0$) in Theorem \ref{th:RVaR} or VaR (when $n=2$) in Theorem \ref{thvar}. 
Notably, this common layered form emerges in our asymptotic setting without requiring any distributional assumptions on $X_i$ or functional form restrictions on $f_i$. 

The proof for Proposition \ref{pro:6} follows along the same lines as that of Proposition \ref{exist} and is therefore omitted.
\begin{proposition}\label{pro:6}
 Let $\widetilde G$ be given by \eqref{eq:tildeG}. The parameters of the optimal ceded loss function     $ {\bf g}_{\a,\b}$  of Theorem \ref{RVaR_as} exists and is determined by
  $$({\bf a}^*, {\bf b}^*)\in \arg\inf_{(\a,\b)\in\mathcal{A}_1} G({\bf a}, {\bf b}).$$
\end{proposition}

\subsection{The result for VaR}
In this section, we solve for the optimal insurance contract under the VaR-based criterion in   the asymptotic framework. Firstly, we show that  solving the VaR-based optimal insurance problem 
\begin{align}\label{eq:var11} V(\mathbf f)=\sum_{i=1}^n \VaR_{\alpha_i}\left(T_{f_i}(X_i)\right)+ \VaR_{\alpha}\left(R_{\mathbf f}(\mathbf X)\right)
\end{align} over $\mathbf f\in \mathcal I^n$ is equivalent to solving \eqref{eq:var11} over $\mathbf g_{\a,\b}\in \mathcal I^n$, as shown in the   proposition below.

\begin{proposition}\label{prop:5} The optimization problem  \eqref{eq:var11} over $\mathcal I^n$ can be equivalently reformulated as 
\begin{align*}
\inf_{\f\in\mathcal{I}^n}  V(\mathbf f) =\inf_{(\a,\b)\in\mathcal A_1}G_0(\a, \b),
\end{align*}
where $G_0({\bf a}, {\bf b})$ is given by \begin{equation}\label{eq:goab}\begin{aligned}   
G_0({\bf a}, {\bf b})=\sum_{i=1}^n \left\{\mathrm{VaR}_{\alpha_i}(X_i)-g_{a_i,b_i}\left(\VaR_{\alpha_i}(X_i)\right)\right\}+\mathrm{VaR}_\alpha\left(\sum_{i=1}^ng_{a_i,b_i}( X_i)\right).\end{aligned}
\end{equation}
\end{proposition}


Proposition \ref{prop:5} indicates that, when the multivariate risks 
$\{X_i\}_{i=1}^n$ are i.i.d., the Central Limit Theorem can be applied 
directly to the aggregate $\sum_{i=1}^n g_{a_i,b_i}$ in the asymptotic 
analysis of aggregation risk, which gives
\begin{equation}\label{pro:VaRas}
\min_{g_{\a,\b}\in \mathcal A_1}  
\left\{
\sum_{i=1}^n \mathrm{VaR}_{\alpha_i}(X - g_{a_i,b_i}(X))
+ \mu(g_{\a,\b})
+ \Phi^{-1}(\alpha)\, \sigma(g_{\a,\b})
\right\},
\end{equation}
where 
\[
\mu(g_{\a,\b}) = \sum_{i=1}^n \mathbb{E}[g_{a_i,b_i}(X)], 
\qquad
\sigma^2(g_{\a,\b}) = \sum_{i=1}^n \mathrm{Var}(g_{a_i,b_i}(X)),
\]
and the parameters satisfy 
\(0 \le a_i \le \VaR_{\alpha_i}(X)\)  
and  
\(b_i = \VaR_{\alpha_i}(X)\).

In the next theorem, we can determine  the optimal ceded loss functions $g_{\a,\b}$ explicitly for \eqref{pro:VaRas}. For convenience, define $$w_i(a_i)= \int_{a_i}^{\VaR_{\alpha_i}(X_i)}S_X(x)\mathrm d x, ~~\text{and}~~ v_i(a_i)= 2\int_{a_i}^{\VaR_{\alpha_i}(X_i)} (x-a_i) S_X(x)\mathrm d x.$$ 
\begin{theorem}\label{thme}  
The optimal indemnity functions ${\mathbf g}^*_{\a,\b}=(g^*_{a_1,b_1},\dots,g^*_{a_n,b_n})$ for Problem \eqref{pro:VaRas} are given by 
\begin{equation}\label{eq:f_star}
g^*_{a^*_i,b^*_i}(x) = (x-a^*_i)_+ - (x-b^*_i)_+, \quad i=1,\dots,n,
\end{equation}
with parameters determined by
\begin{equation*}
a_i^* = \inf \left\{0 \leq a_i \leq \VaR_{\alpha_i}(X_i): 1 - \Phi^{-1}(\alpha) \cdot \frac{w_i(a_i)^2}{\sum_{j=1}^n \left(v_j(a_j) - w_j(a_j)^2 \right)} \geq 0 \right\},
\end{equation*}
and $b_i^* = \VaR_{\alpha_i}(X_i)$.
\end{theorem}

   \begin{corollary}\label{cor}   For $\alpha_1=\dots=\alpha_n$, we have  \begin{equation}\label{eq:f}{g}^*_{a_1,b_1}(x) = \dots={g}^*_{a_n,b_n}(x)  =(x-a^*)_+- (x-\VaR_{\alpha}(X))_+\end{equation}  with 
 $$a^*= \inf \left\{0 \leq a \leq \VaR_{\alpha}(X): 1-\Phi^{-1}(\alpha)\frac{  w(a)^2}{ n (v(a) -w(a)^2 ) }\geq 0    \right\},$$ in which 
$$w(a)= \int_{a}^{\VaR_{\alpha}(X)}S_X(x)\mathrm d x, ~~\text{and}~~ v(a)= 2\int_{a}^{\VaR_{\alpha}(X)} (x-a) S_X(x)\mathrm d x.$$   In particular, $a^*=0$ as $n\to\infty.$ 
 
 \end{corollary} 
 Corollary \ref{cor} highlights that when the insurer pools a large number of i.i.d. risks, the aggregate ceded loss 
$\sum_{i=1}^n g_{a,b}(X_i)$ becomes increasingly predictable due to diversification.  
As a result, the incentive to introduce a positive attachment point $a>0$ vanishes.

\section{Simulation studies}\label{sec:6}
In this section, we present two simulation studies to illustrate our theoretical results. In Section~\ref{sec:6.1}, we provide a case study that examines the minimization of \eqref{pro:VaR} under the special case where RVaR is reduced to VaR. We consider three distinct dependence structures: i.i.d.\ risks, comonotonic risks, and dependence uncertainty. 
Section~\ref{sec:6.2} presents comparative examples that demonstrate the differences in optimal reinsurance design for a general distribution $F$ versus the case where $F \in (\mathcal{M}_{cv})^n$, as characterized in Theorem~\ref{thvar}.

\subsection{Effects of dependence and confidence levels}\label{sec:6.1}

We present an illustrative example that solves the minimization problem (\ref{pro:VaR}) under three different dependence structures---i.i.d., comonotonicity, and dependence uncertainty when $n=2$.   In particular, the optimization problem \eqref{pro:VaR} can be written as \begin{align}\label{eq:K}
    K(\a,\b)= \sum_{i=1}^{2} \VaR_{\alpha_i}(X_i - g_{a_i,b_i}(X_i)) + \VaR_{\alpha}\Big(\sum_{i=1}^{2} g_{a_i,b_i}(X_i)\Big).
\end{align}
Under the different dependence structures, $K(\a,\b) $ takes the following forms: \begin{enumerate}[(i)]
    \item Worst-case: $
        K(\a,\b)  = \inf_{t\in[0,1-\alpha]} \overline G_1(\a, \b, t),
$
where $\overline G_1$ is defined in (\ref{eq:G1bar});
    
    \item  i.i.d. case:  
$
        K(\a,\b) = G_0(\a, \b),
$
    where $G_0$ is defined in (\ref{eq:goab});
    
    \item Comonotonic case: 
$
        K(\a,\b) = \sum_{i=1}^{2} \VaR_{\alpha_i}(X_i - g_{a_i,b_i}(X_i)) + \sum_{i=1}^{2} \VaR_{\alpha}(g_{a_i,b_i}(X_i)).
$
\end{enumerate}
 \begin{lemma}\label{lem:4}
 Suppose $F_1^{-1}$ and $ F_2^{-1}$ are identical and continuous on $(0,1)$, and let $\alpha\in (0,1)$. Then, under the different dependence structures, the following holds:
\begin{align*}
    \inf_{(\a,\b)\in\mathcal A_1}  K(\a,\b) 
    = \inf_{\bf u \in\mathcal A_1(\bf u)} K(\bf u,\bf v)  ,
\end{align*}
where $v_i = \VaR_{\alpha_i}(X_i)$ for $i=1,2$, and
\begin{align*}
    \mathcal {A}_1({\bf u}) = \{{\bf u}: \mathbf{g_{{\bf u}, {\bf v}}} \in \mathcal{I}^n,~0 \le u_i \le v_i \le \infty,~i=1,2\}.
\end{align*}
\end{lemma}
Obtaining a closed-form solution for the optimal insurance problem \eqref{eq:K} remains challenging even under the i.i.d. assumption. Therefore, we employ the asymptotic normality result established in Theorem \ref{thme} to derive numerical solutions. 

We generate independent samples of $X_i$, $i=1,2$, from a Pareto distribution with cumulative distribution function
\[
F_X(x) = 1 - \left(1 + \frac{x}{\lambda}\right)^{-\beta},
\]
where we set $\beta = 9$ and $\lambda = 8$. This parameterization yields $\mathbb{E}[X_i] = 1$ and $\mathrm{Var}(X_i) = \frac{9}{7}$.

We consider three scenarios with $\alpha_1 > \alpha_2$ to examine different configurations of confidence levels:
\begin{itemize}
    \item Case 1: $\alpha_1 = 0.9$, $\alpha_2 = 0.85$, $\alpha = 0.95$
    \item Case 2: $\alpha_1 = 0.95$, $\alpha_2 = 0.85$, $\alpha = 0.9$  
    \item Case 3: $\alpha_1 = 0.95$, $\alpha_2 = 0.9$, $\alpha = 0.85$
\end{itemize}
The corresponding parameter choices and numerical results are summarized in Table~\ref{tab:placeholder}.

\begin{table}[h!]
    \centering
    \caption{Optimal parameters $(a_1^*, a_2^*)$ and objective values under different dependence structures}
    \begin{tabular}{c|l|c|c|c|c}
        \hline
        \textbf{Case} & \textbf{Dependence} & \textbf{Objective Value} & \textbf{$a_1^*$} & \textbf{$a_2^*$} & \textbf{$t^*$} \\ 
        \hline
        \multirow{3}{*}{1} & Worst-case & 4.2096 & $[0, 2.3324]$ & $[0, 1.8772]$ & 0 \\ 
        \cline{2-6}
        & Comonotonic & 4.2096 & $[0, 2.3324]$ & $[0, 1.8772]$ & -- \\ 
        \cline{2-6}
        & i.i.d. & 3.2695 & 0.4224 & 0.3372 & -- \\ 
        \hline
        \multirow{3}{*}{2} & Worst-case & 4.2096 & $[0, 2.3324]$ & $[0, 1.8772]$ & 0 \\ 
        \cline{2-6}
        & Comonotonic & 4.2096 & $[0, 2.3324]$ & $[0, 1.8772]$ & -- \\ 
        \cline{2-6}
        & i.i.d. & 3.1258 & 0.0996 & 0.0072 & -- \\ 
        \hline
        \multirow{3}{*}{3} & Worst-case & 4.2096 & $[0, 1.8772]$ & $[0, 2.3324]$ & 0 \\ 
        \cline{2-6}
        & Comonotonic & 3.7545 & $[0, 1.8772]$ & $[0, 1.8772]$ & -- \\ 
        \cline{2-6}
        & i.i.d. & 2.9832 & 0 & 0 & -- \\ 
        \hline
    \end{tabular}
    \vspace{0.2cm}
    \parbox{\textwidth}{\footnotesize \textbf{Note:}  $\VaR_{0.85}(X) = 1.8772$, $\VaR_{0.9}(X) = 2.3324$. The notation $[0, c]$ indicates that any value in the interval achieves the same optimal objective value, and $t^*$ is related to the optimal value of $t$ in  Proposition \ref{n2}. Dashes (--) indicate that the parameter $t^*$ is not applicable in that setting.} 
    \label{tab:placeholder}
\end{table}

The emergence of interval-valued optimal parameters stems from the specific structure of our objective function $K(\mathbf{a},\mathbf{b})$. From Table \ref{tab:placeholder}, we observe the following patterns: Under the worst-case condition, we obtain $t^*=0$ in all three scenarios, although this does not hold in general (see Section \ref{sec:6.2}). This particular choice of $t^*=0$ implies that the parameter $a^*_2$ can take any value within the interval $[0, \VaR_{\alpha_2}(X)]$ without affecting the optimal objective value, since the condition $\alpha_2 + t \le 1$ is satisfied throughout this range; see Remark \ref{rem:3} for details.

In Case~1, the optimal objective values under both comonotonic and worst-case dependence coincide, equaling $\sum_{i=1}^2 \VaR_{\alpha_i}(X_i)$, which exceeds the value obtained under the i.i.d.\ condition. 
Given that $\alpha \ge \alpha_1 \ge \alpha_2$, the optimal intervals $a_i \in [0, b_i^*]$ indicate that under comonotonic and worst-case dependence, insurers can select any retention level within $[0, \VaR_{\alpha_i}(X_i)]$ for losses up to $\VaR_{\alpha_i}(X_i)$, while retaining all losses beyond this threshold. This results in the same aggregate risk measure as having no reinsurance at all. 

In Case~2, where $\alpha_1 \ge \alpha \ge \alpha_2$, under both comonotonic and worst-case dependence, only the second insurer possesses flexibility, being able to choose any retention level for $a_2^* \in [0, \VaR_{\alpha_2}(X_2)]$ while taking no reinsurance beyond this threshold. This leads to the same outcome as purchasing no reinsurance for the second insurer.
For the first insurer, the optimal strategy involves reinsuring the layer $x - a_1$ for losses in the interval $x \in [\VaR_{\alpha}(X_1), \VaR_{\alpha_1}(X_1)]$, while maintaining flexibility in other regions. Furthermore, under the i.i.d.\ condition, we observe the intuitive result that $a_1 \ge a_2$ when $\alpha_1 \ge \alpha_2$, reflecting that higher risk tolerance requires less reinsurance.
Although the comonotonic scenario yields a higher objective value than the i.i.d.\ case in this particular configuration, we emphasize that this ordering is not universal. As demonstrated in Figure~\ref{fig:va}, the relationship between optimal values under different dependence structures can vary across parameter settings.

    In Case~3, the parameter $t$ remains at $0$, and the optimal reinsurance structure under the worst-case condition resembles that observed in Case~2. The worst-case objective value exceeds the comonotonic value, demonstrating that comonotonic dependence does not always constitute the worst-case scenario for VaR-based risk measures. Under comonotonic dependence,  any combination of $a_1, a_2 \in [0, \VaR_{\alpha}(X_i)]$ minimizes the objective function.
In contrast, under the i.i.d. condition, the unique optimal solution is $a_1 = a_2 = 0$. This outcome follows directly from Theorem~\ref{thme}: when $\alpha$ is sufficiently small relative to $\alpha_1$ and $\alpha_2$, the condition $1 - \Phi^{-1}(\alpha) G_i(\mathbf{a})^{1/2} \ge 0$ holds for all $a_1, a_2 \ge 0$, driving the optimal retention levels to their minimum values.

Next, we assume $\alpha_1 = \alpha_2$ to facilitate a clearer analysis of the relationship between $\alpha$, $\alpha_i$, and the value of the objective function under the optimal  strategy.

\begin{figure}[htb!]
\centering
 \includegraphics[width=14cm]{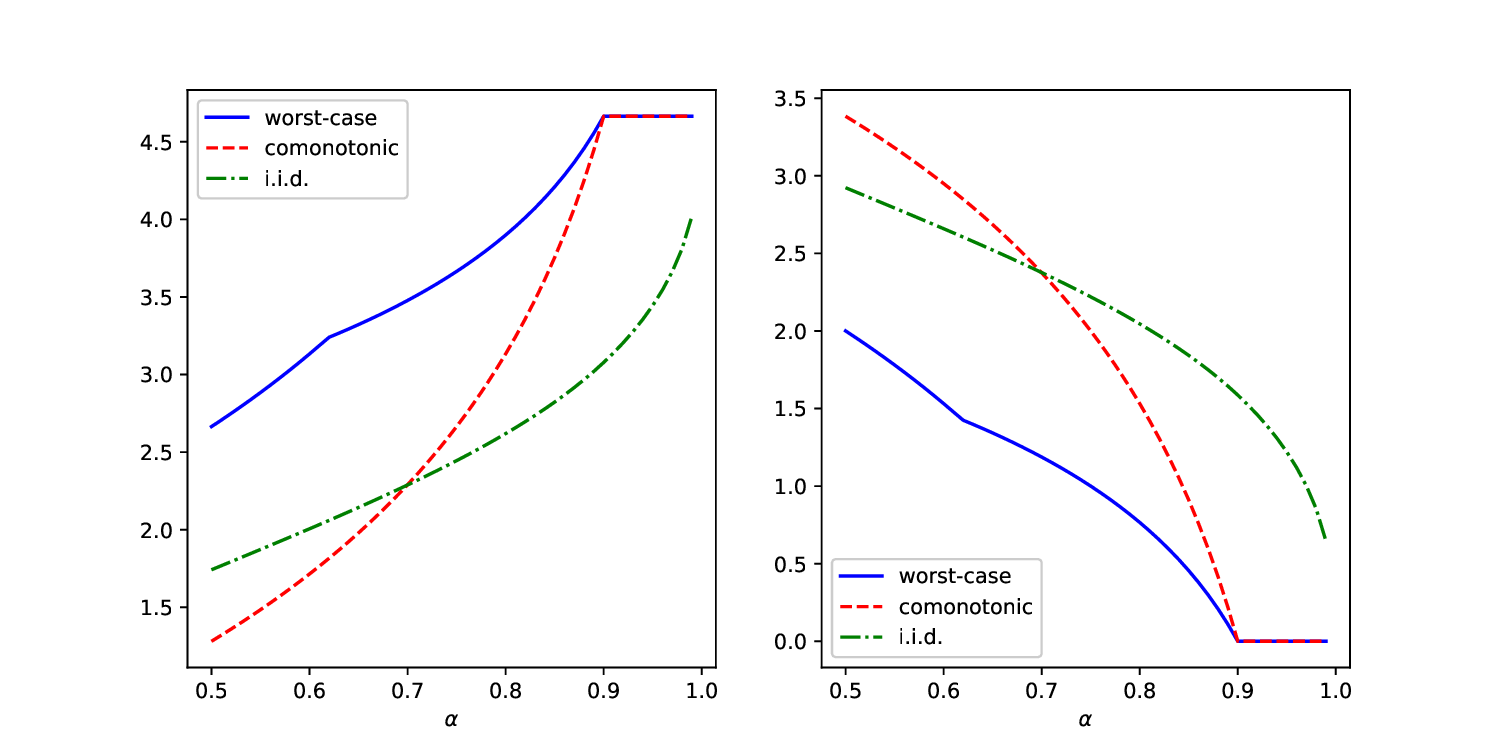}
 \captionsetup{font=small}
\caption{ Value of the objective function under the optimal insurance strategy (left panel) and the benefit of purchasing reinsurance (right panel) for $\alpha\in(0.5,1)$}

 \label{fig:va}
\end{figure}

\begin{figure}[htb!]
\centering
 \includegraphics[width=14cm]{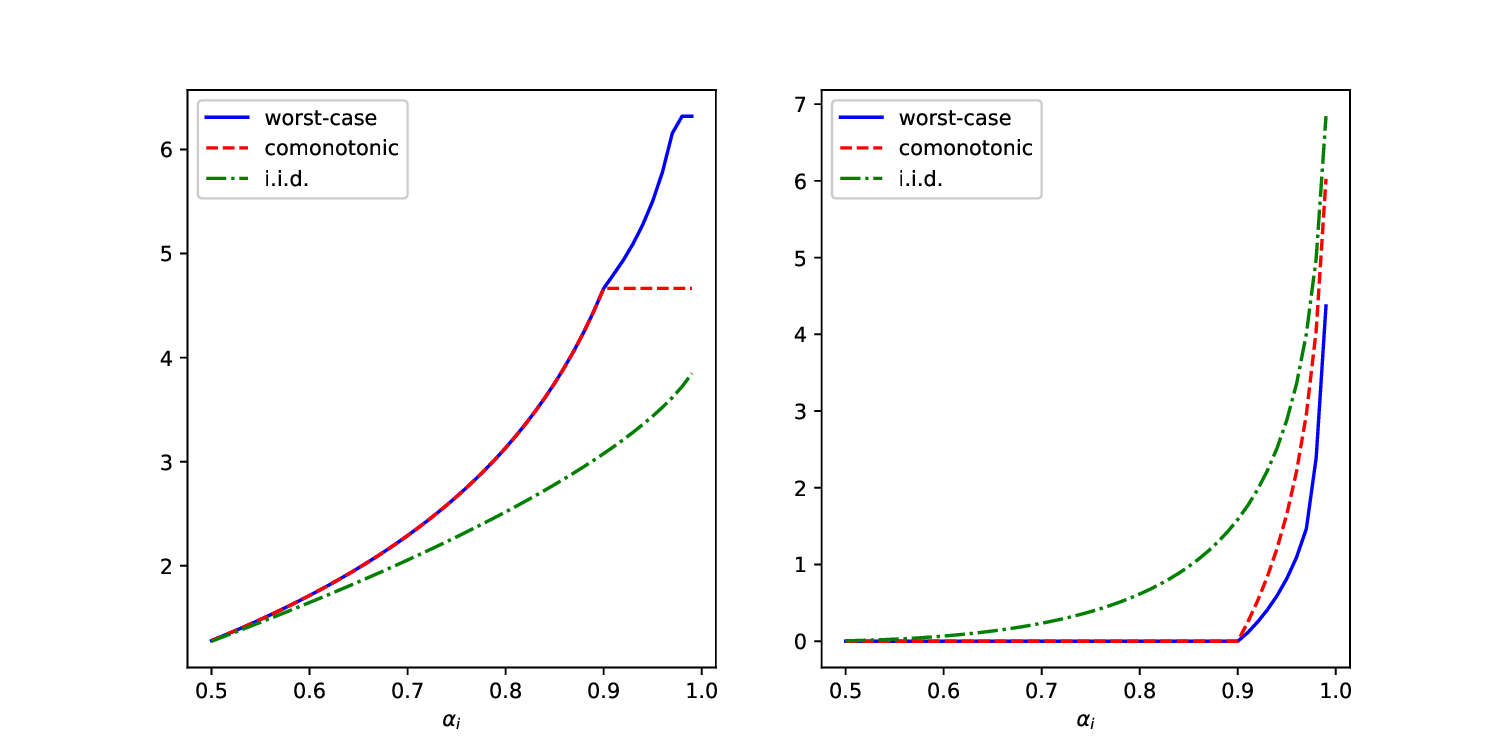}
 \captionsetup{font=small}
 \caption{  Value of the objective function under the optimal insurance strategy (left panel) and the benefit of purchasing reinsurance (right panel) for $\alpha_i\in(0.5,1)$}
 \label{fig:vai}
\end{figure}

As shown in the left panel of Figure~\ref{fig:va}, with fixed $\alpha_1 = \alpha_2 = 0.9$, the optimal objective value increases with $\alpha$. The worst-case scenario consistently yields the highest values, dominating other dependence structures. When $\alpha \ge \alpha_i$, the comonotonic and worst-case results coincide with $\sum_{i=1}^2 \mathrm{VaR}_{\alpha_i}(X_i)$, as evidenced in both Figures~\ref{fig:va} and~\ref{fig:vai}. For sufficiently small $\alpha$, the i.i.d.\ case may outperform the comonotonic scenario, confirming that comonotonic dependence does not always represent the worst-case outcome for VaR-based optimization.

The right panel of Figure~\ref{fig:va} reveals that the benefit of reinsurance — measured as the reduction in the objective value compared to no reinsurance — decreases with $\alpha$. This suggests that the advantage of purchasing reinsurance diminishes when reinsurers employ higher confidence levels. Economically, this aligns with intuition: as reinsurers become more conservative in their risk assessment, the value proposition of reinsurance contracts weakens. Furthermore, when $\alpha \ge \alpha_i$, the optimal objective value coincides with the no-reinsurance case, indicating zero benefit from risk transfer under these parameter conditions.

In Figure~\ref{fig:vai}, with fixed $\alpha = 0.9$, the left panel demonstrates that the optimal objective value increases with $\alpha_i$. While $\mathrm{VaR}_{\alpha_i}(X_i)$ naturally grows with $\alpha_i$, the reinsurance strategy's indirect dependence on this parameter warrants further examination of the net benefit. We therefore analyze the difference:
\begin{align*}
    \sum_{i=1}^2 \mathrm{VaR}_{\alpha_i}(g_{a_i,b_i}(X_i))
    - \mathrm{VaR}_\alpha\!\left(\sum_{i=1}^2 g_{a_i,b_i}(X_i)\right),
\end{align*}
which quantifies the impovement over the no-reinsurance case. As shown in the right panel, this difference increases with $\alpha_i$, indicating that the value of optimal reinsurance becomes more pronounced as insurers adopt higher confidence levels in their risk assessment.   

\subsection{Effects of distributional assumptions}\label{sec:6.2}
In Subsection~\ref{sec:6.1}, we examined specific examples under the worst-case condition for $n=2$ given by Proposition~\ref{n2}, where the optimal parameter $t$ consistently took the value $0$. We now extend our analysis to investigate the general behavior of $t$ and its implications for optimal reinsurance design.

\begin{proposition}\label{a12}
Let $n=2$ and $\overline{G}_1$ be defined by \eqref{eq:G1bar}. Assume that $F_1^{-1}$ and $F_2^{-1}$ are continuous on $(0,1)$, and that $\alpha \ge \alpha_1 + \alpha_2 - 1$. Then, in the optimization problem
\[
\inf_{(a_1,a_2,b_1,b_2)\in\mathcal{A}_1} \inf_{t \in [0,1-\alpha]} \overline{G}_1(a_1,a_2,b_1,b_2,t),
\]  
the optimal value of $t$ is attained at one of the boundary points: either $t^* = 0$ or $t^* = 1-\alpha$.
\end{proposition}

\begin{proposition}\label{aa1a2}
Let $n=2$ and $\overline{G}_1$ be defined by \eqref{eq:G1bar}. Assume that $F_1^{-1}$ and $F_2^{-1}$ are continuous on $(0,1)$, and $\mathbf{F} \in \left(\mathcal{M}_{c x}^\alpha\right)^2$. Then, in the optimization problem
\[
\inf_{(a_1,a_2,b_1,b_2)\in\mathcal{A}_1} \inf_{t\in[0,1-\alpha]} \overline{G}_1(a_1,a_2,b_1,b_2,t),
\]
the optimal value of $t$ is attained at one of the boundary points: either $t^* = 0$ or $t^* = 1-\alpha$.
\end{proposition}

Our theoretical analysis establishes that for $n=2$, the optimal solution reduces to boundary points ($t^* = 0$ or $t^* = 1-\alpha$) under two scenarios: when $\mathbf{F} \in (\mathcal{M}_{\text{cx}}^\alpha)^n$, or when $\alpha_1 + \alpha_2 \le 1 + \alpha$. However, more complex behavior emerges in other parameter configurations. 
We consider the case where $\mathbf{F} \in (\mathcal{M}_{\text{cv}}^\alpha)^n$ and $\alpha_1 + \alpha_2 > 1 + \alpha$. We generate $X_i$ from Pareto distributions with cumulative distribution functions
\[
F_{X_i}(x) = 1 - \left(1 + \frac{x}{\lambda_i}\right)^{-\beta_i}, 
\]
using parameters $(\beta_1, \beta_2) = (9, 6)$ and $(\lambda_1, \lambda_2) = (8, 5)$. This specification yields heterogeneous risk profiles with different tail behaviors. We set the confidence level $\alpha = 0.9$ and report the corresponding optimal parameters $(a^*_1, a^*_2, t^*)$ obtained through numerical optimization in Table~\ref{tab:placeholder2}.

\begin{table}[h!]
    \centering
    \caption{Optimal reinsurance parameters under worst-case dependence with high confidence levels ($\alpha = 0.9$)}
   \begin{tabular}{c|c|c|c|c}
         \hline
         \textbf{Case} & \textbf{result} & \textbf{$a_1^*$} & \textbf{$a_2^*$} & \textbf{$t^*$}\\ \hline
         $\alpha_1=0.97,~\alpha_2=0.99$ & $6.3022$ & $[0,2.3324]$ & $[0,3.9698]$ & $0$ \\ \hline
         $\alpha_1=0.98,~\alpha_2=0.99$ & $6.3944$ & $[0,3.2103]$ & $[0,3.1841]$ & $0.052$ \\ \hline
         $\alpha_1=0.99,~\alpha_2=0.99$ & $6.3944$ & $[0,3.2103]$ & $[0,3.1841]$ & $0.052$ \\ \hline
         $\alpha_1=0.99,~\alpha_2=0.98$ & $6.3944$ & $[0,3.2103]$ & $[0,3.1841]$ & $0.052$ \\ \hline
         $\alpha_1=0.99,~\alpha_2=0.97$ & $6.1503$ & $[0,3.8113]$ & $[0,2.3390]$ & $0.1000$ \\ \hline
    \end{tabular}

    \vspace{0.2cm}
    \parbox{\textwidth}{\footnotesize \textbf{Note:} 
    $\mathrm{VaR}_{0.9}(X_1) = 2.3324$, $\mathrm{VaR}_{0.952}(X_1) = 3.2103$, $\mathrm{VaR}_{0.99}(X_1) = 3.8113$; 
    $\mathrm{VaR}_{0.9}(X_2) = 2.3390$, $\mathrm{VaR}_{0.948}(X_2) = 3.1841$, $\mathrm{VaR}_{0.99}(X_2) = 3.9698$. 
    The notation $[0, c]$ indicates that any retention level in this interval achieves the same optimal objective value, and $t^*$ is related to the optimal value of $t$ in  Proposition \ref{n2}.}
    \label{tab:placeholder2}
\end{table}

Table~\ref{tab:placeholder2} reveals  that boundary solutions persist even when $\alpha_1 + \alpha_2 > 1 + \alpha$: for $(\alpha_1, \alpha_2) = (0.97, 0.99)$ we obtain $t^* = 0$, while for $(0.99, 0.97)$ we have $t^* = 1 - \alpha$. This demonstrates that the reduction to boundary points is not limited to the theoretical condition $\alpha_1 + \alpha_2 \le 1 + \alpha$.

When both $\alpha_1$ and $\alpha_2$ are sufficiently large, the optimal retention levels exhibit the structure $a^*_1 \in [0, \VaR_{\alpha+t^*}(X_1)]$ and $a^*_2 \in [0, \VaR_{1-t^*}(X_2)]$. Notably, the optimal solutions for $(\alpha_1, \alpha_2) = (0.98, 0.99)$, $(0.99, 0.99)$, and $(0.99, 0.98)$ are identical. Similar to Proposition~\ref{aa1a2}, the optimal value $\VaR_{\alpha+t^*}(X_1) + \VaR_{1-t^*}(X_2)$   depends only on  $t^*$ and not on the specific values of $\alpha_1$ and $\alpha_2$. Consequently, both $t^*$ and the optimal objective value remain unchanged across these configurations, as further illustrated in Figure~\ref{fig:vai}.
These findings have practical relevance for insurance markets. Insurers, being typically more risk averse, often employ higher confidence levels (e.g., $\alpha_1, \alpha_2 = 0.99$) compared to reinsurers (e.g., $\alpha = 0.9$). This makes the case $\alpha_1 + \alpha_2 > 1 + \alpha$ particularly relevant in practice, and our results provide guidance for optimal reinsurance design in such realistic settings.

\begin{remark}
It is important to emphasize that our preceding analysis focuses specifically on the case $n = 2$ and addresses the optimal insurance problem under the $\VaR$ risk measure. When considering the more general $\mathrm{RVaR}$-based optimization, the interplay between $\mathrm{RVaR}_{\beta_i,\alpha_i}(X_i - f_i(X_i))$ and $\mathrm{RVaR}_{\gamma_i,\gamma_0}(f_i(X_i))$ introduces additional complexity, necessitating more sophisticated technical treatment.
Consider a simplified setting with $\mathbf{f} \in \mathcal{I}^n$, $\beta = 0$, and $\beta_i = 0$ for $i = 1, 2$. The optimization problem becomes:
\begin{align*}
    \inf_{(a_i,b_i) \in \mathcal{A}_1} \left[ \mathrm{ES}_{1-\alpha_i}(X_i) - \mathrm{ES}_{1-\alpha_i}(g_{a_i,b_i}(X_i)) + \mathrm{ES}_{1-\alpha}(g_{a_i,b_i}(X_i)) \right].
\end{align*}
 If $\alpha \leq \alpha_i$, then
    \begin{align*}
        \mathrm{ES}_{1-\alpha}(g_{a_i,b_i}(X_i)) - \mathrm{ES}_{1-\alpha_i}(g_{a_i,b_i}(X_i)) \geq 0.
    \end{align*}
 If $\alpha > \alpha_i$, we derive:
    \begin{align*}
        & \mathrm{ES}_{1-\alpha}(g_{a_i,b_i}(X_i)) - \mathrm{ES}_{1-\alpha_i}(g_{a_i,b_i}(X_i)) \\
        = & \frac{1}{\alpha} \int_{1-\alpha}^{1-\alpha_i} g_{a_i,b_i}(x)  \mathrm{d}F_X(x) - \left( \frac{1}{\alpha_i} - \frac{1}{\alpha} \right) \int_{1-\alpha_i}^1 g_{a_i,b_i}(x)  \mathrm{d}F_X(x) \\
        = & \frac{1}{\alpha} \int_{1-\alpha}^{1-\alpha_i} \left( g_{a_i,b_i}(x) - g_{a_i,b_i}(\VaR_{1-\alpha_i}(X_i)) \right) \mathrm{d}F_X(x) \\
        & - \left( \frac{1}{\alpha_i} - \frac{1}{\alpha} \right) \int_{1-\alpha_i}^1 \left( g_{a_i,b_i}(x) - g_{a_i,b_i}(\VaR_{1-\alpha_i}(X_i)) \right) \mathrm{d}F_X(x).
    \end{align*}

Note that  for $x \leq \VaR_{1-\alpha_i}(X_i)$, we have $g_{a_i,b_i}(x) - g_{a_i,b_i}(\VaR_{1-\alpha_i}(X_i)) \geq x - \VaR_{1-\alpha_i}(X_i)$, and for $x \geq \VaR_{1-\alpha_i}(X_i)$, we have $g_{a_i,b_i}(x) - g_{a_i,b_i}(\VaR_{1-\alpha_i}(X_i)) \leq x - \VaR_{1-\alpha_i}(X_i)$.
Therefore, the optimal solution for the ES case is achieved with $a^*_i \leq \VaR_{1-\alpha}(X_i)$ and $b^*_i = \infty$, 
which differs from the optimal solution in the $\VaR$ case.
\end{remark}

\section{Conclusion}\label{sec:7}

This paper develops a robust framework for designing Pareto-optimal multilateral reinsurance treaties under dependence uncertainty. Through theoretical analysis and numerical studies, we establish several key insights that advance the understanding of optimal risk sharing in complex insurance markets.

Our main theoretical contribution lies in characterizing the precise structure of Pareto-optimal reinsurance contracts when the dependence between the cedants’ risks is completely unknown. Under the robust RVaR framework (Theorem~\ref{th:RVaR}), the originally infinite-dimensional optimization problem can be reduced to a tractable finite-dimensional one. Optimal indemnities take distinctive parametric forms, including {layered contracts} $g_{a,b}$ that cover losses exceeding a retention level $a$ but are capped at $b$, and {hybrid contracts} $r_{a,b,c}$ that combine proportional and excess-of-loss coverage. These structures allow explicit control over the retained proportion of losses and additional protection above specified thresholds.

For the special case of VaR objectives (Theorem~\ref{thvar}), optimal indemnity functions exhibit  {layer structures} whose forms depend on the convexity properties of marginal distributions. When distributions belong to $\mathcal{M}_{cx}^\alpha$ or $\mathcal{M}_{cv}^\alpha$, we obtain even more explicit characterizations: the {capped proportional contracts} $l_{a,b}$ generalize quota-share arrangements, and the  {shifted excess-of-loss contracts} $h_{a,b}$ include stop loss or quota share as special cases. The asymptotic analysis (Theorems~\ref{RVaR_as} and~\ref{thme}) further demonstrates that as the number of cedants grows, optimal contracts converge to layered structures, represented by $g_{a,b}$.

Our simulation studies highlight several practical patterns. First, the relationship between dependence and optimal outcomes is nuanced: while the worst-case dependence consistently produces conservative outcomes, comonotonicity does not always yield the maximal VaR, and i.i.d.\ scenarios can sometimes generate larger objective values depending on parameter configurations. Second, the value of reinsurance is highly sensitive to confidence levels: higher $\alpha$ in the reinsurer’s assessment reduces the marginal benefit of reinsurance, and higher confidence levels among cedants amplify the advantage of optimal contracts.

Finally, while our analysis assumes {exogenously set premiums}, a natural extension is to study  {premium-dependent strategies}, in which  optimal contracts interact with pricing decisions. Exploring this interaction will be critical for understanding the full economic implications of Pareto-optimal reinsurance design in practice.

\small{}
    \subsection*{Acknowledgments}
The authors are grateful to Yiying Zhang for his helpful comments. Xia Han is supported by the National Natural Science Foundation of China (Grant Nos. 12301604, 12371471, and 12471449).
 \subsection*{Declaration of Interest statements}
The authors declare that they have no known competing financial interests or personal relationships that could have appeared to influence the work reported in this article.

\appendix
\section{Proofs of Section \ref{sec:2}}
\begin{proof}[Proof of Proposition \ref{prop:1}]
We first show the ``if'' part by contradiction. Assume that there exists a contract  $(\mathbf f,\boldsymbol{\pi})  \in\mathcal I^n\times \mathbb{R}^n$ that solves  $\inf_{\mathbf f \in\mathcal I^n}V(\mathbf f),$ but is not robust Pareto-optimal.  This means that there exists $(\hat{\mathbf f},\hat{ \boldsymbol{\pi}}) \in\mathcal I^n\times \mathbb{R}^n$ such that 
\begin{align*}
\mathrm{RVaR}_{\beta_i,\alpha_i}\left(T_{f_i,\pi_i}(X_i)\right)&\geq \mathrm{RVaR}_{\beta_i,\alpha_i}\left(T_{\hat f_i,\hat \pi_i}(X_i)\right),\text{ for all } i\in [n],\\
\sup_{\mathbf X\in \mathcal{E}_n(\mathbf{F})}\mathrm{RVaR}_{\beta,\alpha}\left(R_{\mathbf f,\boldsymbol{\pi}}(\mathbf X)\right) &\geq \sup_{\mathbf X\in \mathcal{E}_n(\mathbf{F})}\mathrm{RVaR}_{\beta,\alpha}\left(R_{\hat{ \mathbf f},\hat{ \boldsymbol{\pi}}}(\mathbf X)\right),
\end{align*}
with at least one strict inequality. It follows that 
\begin{equation*}\begin{aligned}&\sum_{i=1}^n \mathrm{RVaR}_{\beta_i,\alpha_i}\left(T_{f_i,\pi_i}(X_i)\right)+ \sup_{\mathbf X\in \mathcal{E}_n(\mathbf{F})}\mathrm{RVaR}_{\beta,\alpha}\left(R_{\mathbf f,\boldsymbol{\pi}}(\mathbf X)\right)\\&>\sum_{i=1}^n \mathrm{RVaR}_{\beta_i,\alpha_i}\left(T_{\hat f_i,\hat \pi_i}(X_i)\right)+
\sup_{\mathbf X\in \mathcal{E}_n(\mathbf{F})}\mathrm{RVaR}_{\beta,\alpha}\left(R_{\hat{ \mathbf f},\hat{ \boldsymbol{\pi}}}(\mathbf X)\right),\end{aligned}\end{equation*} which contradicts the assumption that $\mathbf f\in \arg\inf_{\mathbf f \in\mathcal I^n}V(\mathbf f).$ Hence, the ``if'' part holds.


To show the ``only if'' part, assume, by way of contradiction, that there exists a robust Pareto-optimal contract $(\mathbf f^*,\boldsymbol{\pi}^*) \in\mathcal I^n\times \mathbb{R}^n$ such that $\mathbf f^*\notin \arg\inf_{\mathbf f \in\mathcal I^n}V(\mathbf f).$ Then, there exists $(\widetilde{\mathbf f},\widetilde{\boldsymbol{\pi}}) \in\mathcal I^n\times \mathbb{R}^n$ such that

\begin{equation}
\begin{aligned}
&\;\sum_{i=1}^{n}\mathrm{RVaR}_{\beta_i,\alpha_i}\left( T_{\widetilde f_i,\widetilde \pi_i}(X_i)\right)+\sup_{\mathbf X\in \mathcal{E}_n(\mathbf{F})}\mathrm{RVaR}_{\beta,\alpha}\left(R_{\widetilde{\mathbf f}, \widetilde{\boldsymbol{\pi}}}(\mathbf X)\right)
\\&<\sum_{i=1}^{n}\mathrm{RVaR}_{\beta_i,\alpha_i}\left( T_{f_i^*,\pi_i^*}(X_i)\right)+\sup_{\mathbf X\in \mathcal{E}_n(\mathbf{F})}\mathrm{RVaR}_{\beta,\alpha}\left(R_{\mathbf f^*,\boldsymbol{\pi}^*}(\mathbf X)\right).
\end{aligned}
\label{eq:star}
\end{equation}
Define, for $i\in [n],$
{
\begin{equation}\begin{aligned}
\hat{\pi}_i:&=\;\widetilde{\pi}_i+\left(\mathrm{RVaR}_{\beta_i,\alpha_i}\left( T_{f_i^*,\pi_i^*}(X_i)\right)-\mathrm{RVaR}_{\beta_i,\alpha_i}\left( T_{\widetilde f_i,\widetilde \pi_i}(X_i)\right)\right)\\&
=\mathrm{RVaR}_{\beta_i,\alpha_i}\left( \widetilde f_i(X_i)\right)-\mathrm{RVaR}_{\beta_i,\alpha_i}\left( f_i^*(X_i)\right)+\pi^*_i.
\label{eq:pi_hat}
\end{aligned}
\end{equation}
}Note that $(\widetilde{\mathbf f},\hat{\boldsymbol{\pi}})\in\mathcal I^n\times \mathbb{R}^n.$ By \eqref{eq:pi_hat} and cash additivity of $\mathrm{RVaR}_{\beta_i,\alpha_i}$,\footnote{For $c\in \mathbb{R}$, $\mathrm{RVaR}_{\beta_i,\alpha_i}(X+c)= \mathrm{RVaR}_{\beta_i,\alpha_i}(X)+c$ holds  for all  $X\in L^1$. } we have, for $i\in [n],$
\begin{equation*}
\mathrm{RVaR}_{\beta_i,\alpha_i}\left( T_{\widetilde f_i,\hat\pi_i}(X_i)\right)=
\mathrm{RVaR}_{\beta_i,\alpha_i}\left( T_{f_i^*,\pi_i^*}(X_i)\right).
\end{equation*}
Moreover, it follows from  \eqref{eq:star}, \eqref{eq:pi_hat}, and cash additivity of $\mathrm{RVaR}_{\beta,\alpha}$ that
\begin{align*}
\sup_{\mathbf X\in \mathcal{E}_n(\mathbf{F})}\mathrm{RVaR}_{\beta,\alpha}\left(R_{\widetilde{\mathbf f},\hat{\boldsymbol{\pi}}}(\mathbf X)\right)&=\sup_{\mathbf X\in \mathcal{E}_n(\mathbf{F})}\mathrm{RVaR}_{\beta,\alpha}\left(R_{\widetilde{ \mathbf f},\widetilde{\boldsymbol{\pi}}}(\mathbf X)\right)+\sum_{i=1}^{n}\mathrm{RVaR}_{\beta_i,\alpha_i}\left( T_{\widetilde f_i,\widetilde \pi_i}(X_i)\right)
\\
&\quad-\sum_{i=1}^{n}\mathrm{RVaR}_{\beta_i,\alpha_i}\left( T_{f_i^*, \pi_i^*}(X_i)\right)\\&<\sup_{\mathbf X\in \mathcal{E}_n(\mathbf{F})}\mathrm{RVaR}_{\beta,\alpha}\left(R_{\mathbf f^*,\boldsymbol{\pi}^*}(\mathbf X)\right).
\end{align*}
 This contradicts the fact that $(\mathbf f^*,\boldsymbol{\pi}^*) \in\mathcal I^n\times \mathbb{R}^n$ is Pareto optimal. 
Hence, the ``only if'' part holds. This completes the proof.
\end{proof}

\section{Proofs of Section \ref{sec:3}}
\begin{proof}[Proof of Theorem \ref{th:RVaR}]  
 (i) 
If  $\beta=0,$  we have $$\begin{aligned}\sup_{\mathbf X \in \mathcal{E}_n(\mathbf F)}\mathrm{RVaR}_{\beta,\alpha}\left(\sum_{i=1}^{n}f_i(X_i)\right)& = \sup_{\mathbf X \in \mathcal{E}_n(\mathbf F)}\mathrm{ES}_{1-\alpha}\left(\sum_{i=1}^{n}f_i(X_i)\right)\\&=\sum_{i=1}^{n}  \ES_{1-\alpha}(f_i(X_i)). \end{aligned}$$ The second   equality holds because ES is subadditive\footnote{For $\alpha\in (0,1]$, $\ES_{1-\alpha}(X+Y)\leq \ES_{1-\alpha}(X)+\ES_{1-\alpha}(Y)$ holds  for all  $X,Y\in L^1$. } and is maximized under the comonotonic dependence structure.
Thus, the optimization problem  \eqref{eq:opt1}  can be written as $$\begin{aligned}\inf_{\mathbf f \in\mathcal I^n}  \sum_{i=1}^n \left\{ \mathrm{RVaR}_{\beta_i,\alpha_i}(X_i) -  \mathrm{RVaR}_{\beta_i,\alpha_i}\left(f_i(X_i)\right)+  \ES_{1-\alpha}(f_i(X_i))\right\}.\end{aligned}$$
By definition, every strategy $ {\bf g}_{{\bf a}, {\bf b}}$ with $({\bf a}, {\bf b})\in \mathcal{A}_1$ belongs to $\mathcal I^n$. Hence, restricting the infimum over $\mathcal I^n$ to this subset yields
\[
\inf_{\mathbf f \in \mathcal I^n} V(\mathbf f)
\le 
\inf_{({\bf a}, {\bf b})\in \mathcal{A}_1}  V( {\bf g}_{{\bf a}, {\bf b}})
= \inf_{ ({\bf a}, {\bf b})\in \mathcal{A}_1} G( {\bf a}, {\bf b}).
\]
 For ease of notation, define $ \bar \alpha_i= 1-\beta_i-\alpha_i,$  $\bar \beta_i =1-\beta_i$ and $\bar \alpha=1-\alpha.$ Next, we explore the optimal indemnity functions under the following three cases (see Figure \ref{fig:ES1}).
\begin{figure}[htb!]
\centering
\includegraphics[width=16cm]{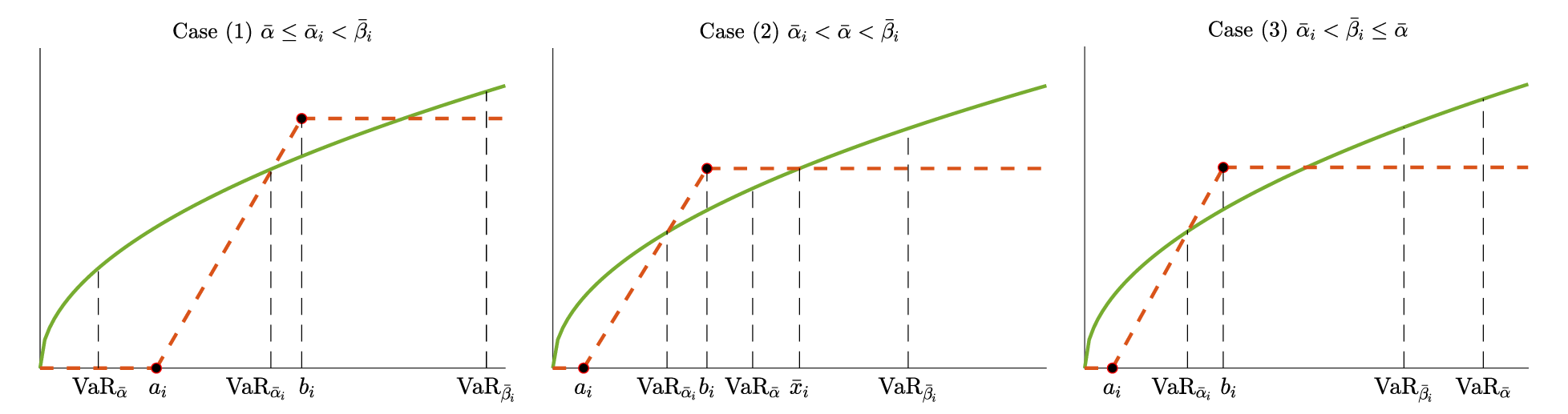}
 \captionsetup{font=small}
 \caption{ \small Three parameter orderings and corresponding $g_{a_i,b_i}$} \label{fig:ES1}
\end{figure}

Let 
\[
a_i = \VaR_{\bar \alpha_i}(X_i) - f_i(\VaR_{\bar \alpha_i}(X_i)), \quad 
\VaR_{\bar \alpha_i}(X_i) \le b_i \le \VaR_{\bar \beta_i}(X_i)
\] 
be such that 
\[
\mathrm{RVaR}_{\beta_i, \alpha_i}\bigl(g_{a_i,b_i}(X_i)\bigr) 
= \mathrm{RVaR}_{\beta_i, \alpha_i}\bigl(f_i(X_i)\bigr).
\] 
The existence of such \(b_i\) is ensured by the continuity of the mapping
\[
t \mapsto \mathrm{RVaR}_{\beta_i, \alpha_i}\bigl(g_{a_i,t}(X_i)\bigr), 
\quad t \in [\VaR_{\bar \alpha_i}(X_i), \VaR_{\bar \beta_i}(X_i)],
\] 
together with the intermediate-value property
\[
\mathrm{RVaR}_{\beta_i, \alpha_i}\bigl(g_{a_i,\VaR_{\bar \alpha_i}(X_i)}(X_i)\bigr)
\le \mathrm{RVaR}_{\beta_i, \alpha_i}\bigl(f_i(X_i)\bigr)
\le \mathrm{RVaR}_{\beta_i, \alpha_i}\bigl(g_{a_i,\VaR_{\bar \beta_i}(X_i)}(X_i)\bigr).
\]

{\bf Case (1)}: $\bar\alpha\leq \bar \alpha_i <\bar \beta_i.$   
With the predefined $a_i,b_i,$ we have $g_{a_i,b_i}(x)\leq f_i(x)$ for $x\leq \VaR_{\bar \alpha_i}(X_i)$ and $x\geq \VaR_{\bar \beta_i}(X_i).$ Hence, we have 
\begin{align*}\ES_{\bar\alpha} (g_{a_i,b_i}(X_i))&=\frac{1}{\alpha}\left(\int_{[\bar\alpha,\bar\alpha_i]\cup [\bar\beta_i, 1]}g_{a_i,b_i}(\VaR_t(X_i))\mathrm d t+\alpha_i\mathrm{RVaR}_{\beta_i,  \alpha_i} (g_{a_i,b_i}(X_i))\right)\\
&\leq \frac{1}{\alpha}\left(\int_{[\bar\alpha,\bar\alpha_i]\cup [\bar\beta_i, 1]}f_i(\VaR_t(X_i))\mathrm d t+\alpha_i\mathrm{RVaR}_{\beta_i,  \alpha_i} (f_i(X_i))\right)\\
&=\ES_{\bar\alpha} (f_i(X)).
\end{align*}  
This implies  that for any $f_i\in \mathcal{I},$ we can always find $g_{a_i,b_i}\in \mathcal{I}$   that is better than $f_i.$

{\bf Case (2)}: $\bar \alpha_i< \bar\alpha <\bar \beta_i.$  By the construction of $a_i$ and $b_i,$ there exists $\VaR_{\bar\alpha_i}(X_i)\leq \bar x_i\leq \VaR_{\bar\beta_i}(X_i)$ such that 
$g_{a_i,b_i}(x)\geq f_i(x)$ for $x\in [\VaR_{\bar\alpha_i}(X_i), \bar x_i]$ and $g_{a_i,b_i}(x)\leq f_i(x)$ for $x\geq \bar x_i.$ Hence, $\mathrm{RVaR}_{\beta_i,  \alpha_i} (g_{a_i,b_i}(X_i)) =\mathrm{RVaR}_{\beta_i,  \alpha_i}(f_i(X_i))$ implies   $$\int_{[\bar\alpha, \bar\beta_i]}g_{a_i,b_i}(\VaR_t(X_i))\mathrm d t\leq \int_{[\bar\alpha, \bar\beta_i]}f_i(\VaR_t(X_i))\mathrm d t,$$
which further implies $\ES_{\bar\alpha} (g_{a_i,b_i}(X_i))\leq \ES_{\bar\alpha} (f_i(X)).$ Therefore, for any $f_i\in \mathcal{I},$ there  always exists $g_{a_i,b_i}\in \mathcal{I}$  such  that it is better than $f_i.$ 

{\bf Case (3)}: $ \bar \alpha_i <\bar \beta_i\leq \bar\alpha.$  In this case, it is obvious that $\ES_{\bar\alpha} (g_{a_i,b_i}(X_i))\leq \ES_{\bar\alpha} (f_i(X)).$ Therefore, for any $f_i\in \mathcal{I},$ there  always exists $g_{a_i,b_i}\in \mathcal{I}$   such that it is better than $f_i\in \mathcal{I}.$ 

To summarize the three cases, we have 
$$\begin{aligned}\inf_{\f\in\mathcal{I}^n}V(\mathbf f)&\geq 
\inf_{(\a,\b)\in\mathcal{A}_1} G(\a,\b) .\end{aligned}$$
 We obtain the desired result for (i).

(ii)  
Let $f_i \in \mathcal{I}$ be convex, and consider $X_i$ with $F_{X_i} \in \mathcal{M}_{cv}^{1-\beta-\alpha}$.  
To show $F_{f_i(X_i)} \in \mathcal{M}_{cv}^{1-\beta-\alpha}$, take any $y_1, y_2 \ge f_i(F_+^{-1}(1-\beta-\alpha))$ with $y_1 \le y_2$, and let $x_1 = f_i^{-1}(y_1)$, $x_2 = f_i^{-1}(y_2)$.  
By the monotonicity of $f_i$, the preimage of the interval $[y_1,y_2]$ lies in $[x_1,x_2]$.  
Then, for any $\lambda \in (0,1)$, by the convexity of $f_i$,
\[
f_i(\lambda x_1 + (1-\lambda) x_2) \le \lambda f_i(x_1) + (1-\lambda) f_i(x_2) = \lambda y_1 + (1-\lambda) y_2,
\]
which implies
\[
f_i^{-1}(\lambda y_1 + (1-\lambda)y_2) \ge \lambda x_1 + (1-\lambda)x_2.
\]
Since $F_{X_i} \in \mathcal{M}_{cv}^{1-\beta-\alpha}$, for $x_1,x_2 \ge F_+^{-1}(1-\beta-\alpha)$ we have
\[
F_{X_i}(\lambda x_1 + (1-\lambda)x_2) \ge \lambda F_{X_i}(x_1) + (1-\lambda) F_{X_i}(x_2).
\]
Combining the above inequalities and using monotonicity of $f_i$, we obtain
\[
\begin{aligned}
F_{f_i(X_i)}(\lambda y_1 + (1-\lambda)y_2) 
&= F_{X_i}(f_i^{-1}(\lambda y_1 + (1-\lambda)y_2)) \\
&\ge F_{X_i}(\lambda x_1 + (1-\lambda)x_2) \\
&\ge \lambda F_{X_i}(x_1) + (1-\lambda) F_{X_i}(x_2) \\
&= \lambda F_{f_i(X_i)}(y_1) + (1-\lambda) F_{f_i(X_i)}(y_2),
\end{aligned}
\]
which proves that $F_{f_i(X_i)} \in \mathcal{M}_{cv}^{1-\beta-\alpha}$.

 Hence, in light of  Lemma \ref{lem:0}, the optimization problem  \eqref{eq:opt1} becomes 
\begin{equation}\label{eq:opt2}
\begin{aligned}
&\inf_{\mathbf f \in\mathcal I_{cx}^n} \inf_{\pmb{\gamma}\in(\beta+\alpha)\Delta_n,\gamma_0\geq \alpha} \left\{\sum_{i=1}^n \mathrm{RVaR}_{\beta_i,\alpha_i}\left(T_{f_i}(X_i)\right)+    \sum_{i=1}^n   \mathrm{RVaR}_{\gamma_i,\gamma_0}(f_i(X_i))\right\}\\  =&  \inf_{\pmb{\gamma}\in(\beta+\alpha)\Delta_n,\gamma_0\geq \alpha}  \inf_{\mathbf f \in\mathcal I_{cx}^n}  \sum_{i=1}^n \left\{ \mathrm{RVaR}_{\beta_i,\alpha_i}(X_i) -  \mathrm{RVaR}_{\beta_i,\alpha_i}\left(f_i(X_i)\right)+  \mathrm{RVaR}_{\gamma_i,\gamma_0}(f_i(X_i))\right\}. 
\end{aligned}
\end{equation}
By a proof similar to Theorem \ref{th:RVaR} (i), we have $$\begin{aligned}\inf_{\f\in\mathcal{I}_{cx}^n}V(\mathbf f)\leq \inf_{(\a,\b,\c)\in\mathcal{A}_2} \inf_{\pmb{\gamma}\in(\beta+\alpha)\Delta_n,\gamma_0\geq \alpha}R(\a,\b,\c,\pmb{\gamma}) .\end{aligned}$$
 
Let $ \bar \gamma_0= 1-\gamma_0-\gamma_i$ and $\bar \gamma_i =1-\gamma_i.$  We need to show the inverse inequality for \eqref{eq:opt2} under the following six cases (see Figure \ref{fig: RVaR}). 
\begin{figure}[htb!]
\centering
 \includegraphics[width=13cm]{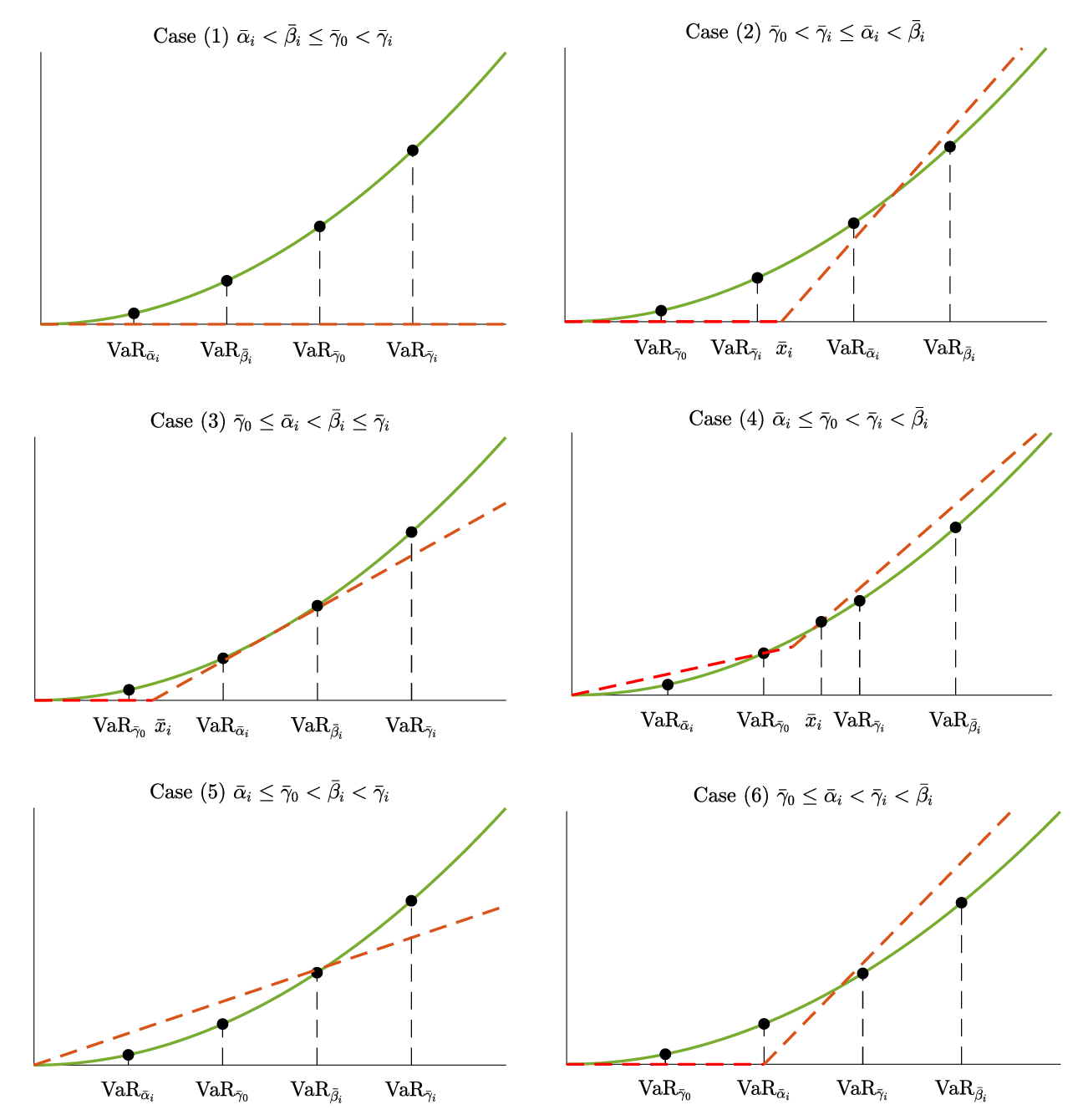}
 \captionsetup{font=small}
\caption{\small Six parameter orderings and corresponding    $r_{a_i,b_i,c_i}$,  each outperforming convex indemnities.}
 \label{fig: RVaR}
\end{figure}

{\bf Case (1): $\bar \alpha_i <\bar \beta_i \leq \bar \gamma_0<\bar\gamma_i$.} 
In this case, we have $\mathrm{RVaR}_{\beta_i,\alpha_i}\left(f_i(X_i)\right)\leq   \mathrm{RVaR}_{\gamma_i,\gamma_0}(f_i(X_i)) $  for any $f_i\in\mathcal  I_{cx}.$ Thus we have  $\mathrm{RVaR}_{\gamma_i,\gamma_0}(f_i(X_i))-\mathrm{RVaR}_{\beta_i,\alpha_i}\left(f_i(X_i)\right) \geq 0$, which implies that   $r_{0,0,0}\in \mathcal I_{cx}$ always   performs better than $f_i.$

 {\bf Case (2): $\bar \gamma_0 < \bar \gamma_i  \leq \bar \alpha_i   <\bar \beta_i $.}  
  Define $a_i= 0$, $b_i=\bar x_i$ and $c_i=1,$  where $\VaR_{\bar \alpha_i}(X_i)-f_i(\VaR_{\bar \alpha_i}(X_i)) \leq \bar x_i \leq \VaR_{\bar \beta_i}(X_i),$  such that  $\mathrm{RVaR}_{\beta_i,  \alpha_i}  (r_{a_i,b_i, c_i}(X_i))=\mathrm{RVaR}_{\beta_i,  \alpha_i}(f_i(X_i)).$ The existence of such $\bar x_i$ is guaranteed by the continuity of $\mathrm{RVaR}_{\beta_i,  \alpha_i} (r_{a_i,t, c_i}(X_i))$ for $t\in [\VaR_{\bar \alpha_i}(X_i)-f_i(\VaR_{\bar \alpha_i}(X_i)), \VaR_{\bar \beta_i}(X_i)]$ and the fact that $$\mathrm{RVaR}_{\beta_i,  \alpha_i} (r_{a_i,\VaR_{\bar \beta_i}(X_i), c_i}(X_i))\leq \mathrm{RVaR}_{\beta_i,  \alpha_i}(f_i(X_i))\leq \mathrm{RVaR}_{\beta_i,  \alpha_i} (r_{a_i,\VaR_{\bar \alpha_i}(X_i)-f_i(\VaR_{\bar \alpha_i}(X_i)), c_i}(X_i)).$$ 
The convexity of $f_i$ implies that $r_{a_i,b_i, c_i}(x)\leq f_i(x)$ for $x\leq \bar x.$ Hence, we have  $ \mathrm{RVaR}_{\gamma_i,  \gamma_0}(r_{a_i,b_i, c_i}(X_i))  \leq \mathrm{RVaR}_{\gamma_i,  \gamma_0}(f_i(X_i)),$ which implies $\mathrm{RVaR}_{\gamma_i,\gamma_0}(f_i(X_i))-\mathrm{RVaR}_{\beta_i,\alpha_i}\left(f_i(X_i)\right)\geq\mathrm{RVaR}_{\gamma_i,\gamma_0}(r_{a_i,b_i, c_i}(X_i))-\mathrm{RVaR}_{\beta_i,\alpha_i}\left(r_{a_i,b_i, c_i}(X_i)\right).$ 

 {\bf Case (3):  $  \bar \gamma_0 \leq \bar \alpha_i    <\bar \beta_i\leq \bar \gamma_i$.} 
 Let $$k_{1i}=\frac{f_i(\VaR_{\bar \beta_i}(X_i))-f_i(\VaR_{\bar \alpha_i} (X_i))}{\VaR_{\bar \beta_i} (X_i)-\VaR_{\bar \alpha_i} (X_i) },$$ with the convention that $\frac{0}{0}=0.$
For any $f\in \mathcal I_{cx},$ let $ a_i= 0$ ,  $b_i=\bar x_i$ and $c_i=k_{1i},$ where  $$\VaR_{\bar \alpha_i}(X_i)-\frac{1}{k_{1i}}f_i(\VaR_{\bar \alpha_i} (X_i))) \leq  \bar x_i \leq  \VaR_{\bar \beta_i}(X_i),$$     such that 
$\mathrm{RVaR}_{\beta_i,  \alpha_i} (r_{a_i,b_i,c_i}(X_i))) =\mathrm{RVaR}_{\beta_i,  \alpha_i}(f_i(X_i)).$ 
 The existence of such $\bar x_i$ is guaranteed by the continuity of $\mathrm{RVaR}_{\beta_i,  \alpha_i} (r_{a_i,t, c_i}(X_i))$ for $t\in [\VaR_{\bar \alpha_i}(X_i)-\frac{1}{k_{1i}}f_i(\VaR_{\bar \alpha_i} (X_i))) ,\VaR_{\bar \beta_i}(X_i)],$
 and the fact that $$\mathrm{RVaR}_{\beta_i,  \alpha_i} (r_{a_i, \VaR_{\bar \beta_i}(X_i), c_i}(X_i))\leq \mathrm{RVaR}_{\beta_i,  \alpha_i}(f_i(X_i))\leq \mathrm{RVaR}_{\beta_i,  \alpha_i} (r_{a_i,\VaR_{\bar \alpha_i}(X_i)-\frac{1}{k_{1i}}f_i(\VaR_{\bar \alpha_i} (X_i))),c_i}(X_i)).$$ 
By the convexity of $f_i,$ we have $r_{a_i,b_i,c_i}(x)\leq f_i(x)$ for $x\leq \VaR_{\bar \alpha_i}$ and $x\geq \VaR_{\bar \beta_i}.$ So, we have
$$\begin{aligned}\mathrm{RVaR}_{\gamma_i,  \gamma_0} (r_{a_i,b_i,c_i}(X_i)) &=\frac 1\gamma_0  \left(\int_{[\bar\gamma_0,\bar\alpha_i]\cup [\bar\beta_i, \bar \gamma_i]}r_{a_i,b_i,c_i}(\VaR_t(X_i))\mathrm d t+\alpha_i\mathrm{RVaR}_{\beta_i,  \alpha_i} (r_{a_i,b_i,c_i}(X_i))\right)&\\ 
&\leq \frac 1 \gamma_0\left(\int_{[\bar\alpha,\bar\alpha_i]\cup [\bar\beta_i, 1]}f_i(\VaR_t(X_i))\mathrm d t+\alpha_i\mathrm{RVaR}_{\beta_i,  \alpha_i} (f_i(X_i))\right)\\
&=\mathrm{RVaR}_{\gamma_i,  \gamma_0}   (f_i(X)).\end{aligned}$$   This implies  that we can always find $r_{a_i,b_i,c_i}\in \mathcal I_{cx}$   that is better than $f_i.$ 

  {\bf Case (4):  $\bar \alpha_i \leq \bar \gamma_0 < \bar \gamma_i<\bar \beta_i$.} 
  Let $ k_{2i}=  (f_i)_-'(\VaR_{\bar \beta_i} (X_i)),$ where $(f_i)_-'(x)$ is the left derivative of $f_i$ at $x$. Define  $ a_i= \frac{f_i(\VaR_{\bar\gamma_0}(X_i))}{\VaR_{\bar\gamma_0}(X_i)},$  $b_i=\bar x_i$ and    
  $c_i= k_{2i}-a_i,$   where $$ \VaR_{\bar \gamma_0}(X_i)\leq\bar x_i\leq\VaR_{\bar\gamma_i}(X_i)$$ such that 
  $ \mathrm{RVaR}_{\gamma_i,  \gamma_0} ( r_{a_i,b_i,c_i}(X_i))=\mathrm{RVaR}_{\gamma_i,  \gamma_0} ( f_i(X_i)).$ 
 The existence of such $\bar x_i$ is guaranteed by the continuity of $\mathrm{RVaR}_{\gamma_i,  \gamma_0} (r_{a_i,t, c_i}(X_i))$ for $t\in [\VaR_{\bar \gamma_0}(X_i) ,\VaR_{\bar\gamma_i}(X_i)]$ 
 and the fact that $$\mathrm{RVaR}_{\gamma_i,  \gamma_0} (r_{a_i, \VaR_{\bar\gamma_i}(X_i), c_i}(X_i))\leq \mathrm{RVaR}_{\gamma_i,  \gamma_0}(f_i(X_i))\leq \mathrm{RVaR}_{\gamma_i,  \gamma_0} (r_{a_i,\VaR_{\bar\gamma_0}(X_i),c_i}(X_i)).$$  
 The convexity of  $f_i$ implies that $r_{a_i,b_i,c_i}(x)\geq f_i(x)$ for $\VaR_{\bar \alpha_i}(X_i)\leq x \leq \VaR_{\bar \gamma_0}(X_i)$ and $\VaR_{\bar \gamma_i}(X_i)\leq x \leq \VaR_{\bar \beta_i}(X_i).$ Hence, we have
$$\begin{aligned}\mathrm{RVaR}_{\beta_i,  \alpha_i} (r_{a_i,b_i,c_i}(X_i)) &=\frac{1}{\alpha_i}  \left(\int_{[\bar\alpha_i,\bar\gamma_0]\cup [ \bar \gamma_i, \bar\beta_i]}r_{a_i,b_i,c_i}(\VaR_t(X_i))\mathrm d t+\gamma_0\mathrm{RVaR}_{\gamma_i,  \gamma_0} (r_{a_i,b_i,c_i}(X_i))\right)\\ 
&\geq  \frac{1}{\alpha_i} \left(\int_{[\bar\alpha_i,\bar\gamma_0]\cup [ \bar \gamma_i, \bar\beta_i]} f_i(\VaR_t(X_i))\mathrm d t+\gamma_0\mathrm{RVaR}_{\gamma_i,  \gamma_0} (f_i(X_i))\right)\\
&=\mathrm{RVaR}_{\beta_i,  \alpha_i}   (f_i(X_i)).\end{aligned}$$
Consequently,  $\mathrm{RVaR}_{\gamma_i,\gamma_0}(f_i(X_i))-\mathrm{RVaR}_{\beta_i,\alpha_i}\left(f_i(X_i)\right)\geq\mathrm{RVaR}_{\gamma_i,\gamma_0}(r_{a_i,b_i, c_i}(X_i))-\mathrm{RVaR}_{\beta_i,\alpha_i}\left(r_{a_i,b_i, c_i}(X_i)\right).$
 
 {\bf Case (5): $\bar \alpha_i \leq \bar \gamma_0 <\bar \beta_i  <\bar\gamma_i$.} 
Define  $a_i=\frac{f_i(\bar x_i)}{\bar x_i}$ and $b_i=c_i=0,$  where $\VaR_{\bar \gamma_0}(X_i) \leq \bar x_i \leq  \VaR_{\bar \gamma_i}(X_i),$  such that  $ \mathrm{RVaR}_{\gamma_i,  \gamma_0}(r_{a_i,b_i,c_i}(X_i))=\mathrm{RVaR}_{\gamma_i,  \gamma_0}(f_i(X_i)). $ 
 The existence of such $\bar x_i$ is guaranteed by the continuity of $\mathrm{RVaR}_{\gamma_i,  \gamma_0} (r_{t,b_i, c_i}(X_i))$ for $t\in [\frac{f_i(\VaR_{\bar \gamma_0}(X_i))}{\VaR_{\bar \gamma_0}(X_i)} , \frac{f_i(\VaR_{\bar \gamma_i}(X_i))}{\VaR_{\bar \gamma_i}(X_i)}],$
 and the fact that $$\mathrm{RVaR}_{\gamma_i,  \gamma_0} (r_{f_i(\VaR_{\bar \gamma_0}(X_i))/\VaR_{\bar \gamma_0}(X_i), b_i,c_i}(X_i))\leq \mathrm{RVaR}_{\gamma_i,  \gamma_0}(f_i(X_i))\leq \mathrm{RVaR}_{\gamma_i,  \gamma_0} (r_{f_i(\VaR_{\bar \gamma_i}(X_i))/\VaR_{\bar \gamma_i}(X_i),b_i,c_i}(X_i)).$$  
Note that $ \mathrm{RVaR}_{\gamma_i,  \gamma_0}(r_{a_i,b_i,c_i}(X_i))=\mathrm{RVaR}_{\gamma_i,  \gamma_0}(f_i(X_i))$ implies  $$\int_{[\bar \gamma_0,\bar\beta_i]} r_{a_i,b_i,c_i}(\VaR_t(X_i))\mathrm d t\geq\int_{[\bar \gamma_0,\bar\beta_i]} f_i(\VaR_t(X_i))\mathrm d t.$$
By the convexity of $f_i,$ we have $r_{a_i,b_i,c_i}(x)\geq f_i(x)$ for  $\VaR_{\bar \alpha_i}(X_i)\leq x \leq \VaR_{\bar \gamma_0}(X_i).$ Hence, we have
$$\begin{aligned}&\mathrm{RVaR}_{\beta_i,  \alpha_i} (r_{a_i,b_i,c_i}(X_i))\\ &=\frac{1}{\alpha_i}  \left(\int_{[\bar\alpha_i,\bar \gamma_0]}r_{a_i,b_i,c_i}(\VaR_t(X_i))\mathrm d t+\int_{[\bar \gamma_0,\bar\beta_i]}r_{a_i,b_i,c_i}(\VaR_t(X_i))\mathrm d t\right)\\
&\geq  \frac{1}{\alpha_i}  \left(\int_{[\bar\alpha_i,\bar \gamma_0]}f_i(\VaR_t(X_i))\mathrm d t+\int_{[\bar \gamma_0,\bar\beta_i]}f_i(\VaR_t(X_i))\mathrm d t\right)\\
&=\mathrm{RVaR}_{\beta_i,  \alpha_i}   (f_i(X)),\end{aligned}$$   implying  that we can always find $r_{a_i,b_i,c_i}\in \mathcal I_{cx}$   that is better than $f_i.$

{\bf Case (6): $  \bar \gamma_0 \leq \bar \alpha_i  < \bar \gamma_i <\bar \beta_i$.}  
  Define $a_i=0$, $b_i=\bar x_i$ and  $c_i=1,$  where $\VaR_{\bar \alpha_i}(X_i)-f_i(\VaR_{\bar \alpha_i}(X_i)) \leq \bar x_i \leq \VaR_{\bar \beta_i}(X_i),$  such that  $\mathrm{RVaR}_{\beta_i,  \alpha_i}  (r_{a_i,b_i, c_i}(X_i))=\mathrm{RVaR}_{\beta_i,  \alpha_i}(f_i(X_i)).$ The existence of such $\bar x_i$ is guaranteed by the continuity of $\mathrm{RVaR}_{\beta_i,  \alpha_i} (r_{a_i,t, c_i}(X_i))$ for $t\in [\VaR_{\bar \alpha_i}(X_i)-f_i(\VaR_{\bar \alpha_i}(X_i)),\VaR_{\bar \beta_i}(X_i)]$ and the fact that $$\mathrm{RVaR}_{\beta_i,  \alpha_i} (r_{a_i,\VaR_{\bar \beta_i}(X_i), c_i}(X_i))\leq \mathrm{RVaR}_{\beta_i,  \alpha_i}(f_i(X_i))\leq \mathrm{RVaR}_{\beta_i,  \alpha_i} (r_{a_i,\VaR_{\bar \alpha_i}(X_i)-f_i(\VaR_{\bar \alpha_i}(X_i)), c_i}(X_i)).$$ 
  Note that $\mathrm{RVaR}_{\beta_i,  \alpha_i}  (r_{a_i,b_i, c_i}(X_i))=\mathrm{RVaR}_{\beta_i,  \alpha_i}(f_i(X_i))$ implies $$\int_{[\bar \alpha_i,\bar\gamma_i]} r_{a_i,b_i,c_i}(\VaR_t(X_i))\mathrm d t\leq\int_{[\bar \alpha_i,\bar\gamma_i]} f_i(\VaR_t(X_i))\mathrm d t.$$
By the convexity of $f_i$, we have $r_{a_i,b_i,c_i}(x)\leq f_i(x)$ for  $\VaR_{\bar \gamma_0}(X_i)\leq x \leq \VaR_{\bar \alpha_i}(X_i).$ Hence, we have
$$\begin{aligned}
&\mathrm{RVaR}_{\gamma_i,  \gamma_0} (r_{a_i,b_i,c_i}(X_i))\\ &=\frac{1}{\gamma_0}  \left(\int_{[\bar \gamma_0,\bar\alpha_i]}r_{a_i,b_i,c_i}(\VaR_t(X_i))\mathrm d t+\int_{[\bar \alpha_i,\bar\gamma_i]}r_{a_i,b_i,c_i}(\VaR_t(X_i))\mathrm d t\right)\\
&\leq  \frac{1}{\gamma_0}  \left(\int_{[\bar \gamma_0,\bar\alpha_i]}f_i(\VaR_t(X_i))\mathrm d t+\int_{[\bar \alpha_i,\bar\gamma_i]}f_i(\VaR_t(X_i))\mathrm d t\right)\\
&=\mathrm{RVaR}_{\gamma_i,  \gamma_0}   (f_i(X)).
\end{aligned}$$  
Consequently, we have  $$\mathrm{RVaR}_{\gamma_i,\gamma_0}(f_i(X_i))-\mathrm{RVaR}_{\beta_i,\alpha_i}\left(f_i(X_i)\right)\geq\mathrm{RVaR}_{\gamma_i,\gamma_0}(r_{a_i,b_i, c_i}(X_i))-\mathrm{RVaR}_{\beta_i,\alpha_i}\left(r_{a_i,b_i, c_i}(X_i)\right).$$
To summarize the six cases, we have 
$$\begin{aligned}\inf_{\f\in\mathcal{I}_{cx}^n}V(\mathbf f)&\geq 
\inf_{\pmb{\gamma}\in(\beta+\alpha)\Delta_n,\gamma_0\geq \alpha}\inf_{(\a,\b,\c)\in\mathcal{A}_2} R(\a,\b,\c,\pmb{\gamma})\\&= \inf_{(\a,\b,\c)\in\mathcal{A}_2} \inf_{\pmb{\gamma}\in(\beta+\alpha)\Delta_n,\gamma_0\geq \alpha}R(\a,\b,\c,\pmb{\gamma}) .\end{aligned}$$
  We obtain the desired result.
\end{proof}
\begin{proof}[Proof of Proposition \ref{exist}] It is evident that ${\bf g}_{\a^*,\b^*}$ and  ${\bf r}_{\a^*,\b^*,\c^*}$ are the optimal ceded loss functions for cases (i)-(ii), respectively.   The existence of $(\a^*,\b^*)\in \mathcal A_1$ and $(\a^*,\b^*,\c^*)\in \mathcal A_2$ follows from the continuity of the functions $G$ and $R$. Specifically, the existence of $(\a^*,\b^*)$ is guaranteed by the continuity of $\mathrm{RVaR}$ and $\mathrm{ES}$, while the existence of $(\a^*,\b^*,\c^*)$ needs an additional argument because $\mathrm{RVaR}_{\gamma_i,\gamma_0}(X_i)$ may lose continuity as $\gamma_0\downarrow 0$.%


To be specific, for $0\leq \epsilon\leq T\leq 1,$ define $$\bar\Delta_n^{\epsilon,T}=\left\{(\gamma_0,\gamma_1,\ldots,\gamma_n)\in[\epsilon,T]\times[0,T]^n:\sum^n_{i=0}\gamma_i=T\right\}.$$
Clearly, for $0<\epsilon<\beta+\alpha$,  $\mathrm{RVaR}_{\gamma_i,\gamma_0}(X_i)$ is continuous with respect to $\pmb\gamma$ over $\bar\Delta_n^{\epsilon,\beta+\alpha}$, and $$\mathrm{RVaR}_{\gamma_i,\gamma_0}(r_{a_i,b_i,c_i}(X_i))=\frac{1}{\gamma_0}\int^{\gamma_i+\gamma_0}_{\gamma_i}r_{a_i,b_i,c_i}(F^{-1}_i(t))\mathrm dt,$$  which implies that $\mathrm{RVaR}_{\gamma_i,\gamma_0}(r_{a_i,b_i,c_i}(X_i))$ is continuous with respect to $(\a,\b,\c,\pmb\gamma)$ over $\mathcal A_2\times\bar\Delta_n^{\epsilon,\beta+\alpha}$ for $\epsilon>0.$ So, by the continuity of $\mathrm{RVaR},$ $R({\bf a}, {\bf b},\c, \pmb\gamma)$ is a continuous function of $({\bf a}, {\bf b},\c, \pmb\gamma)$ over $\mathcal A_2\times\bar\Delta_n^{\epsilon,\beta+\alpha}$ for $\epsilon>0$ and $\mathcal A_2$ is a closed set. Therefore, there exists $(\a^*_\epsilon,\b^*_\epsilon,\c^*_\epsilon,\pmb\gamma^*_\epsilon)\in\mathcal A_2\times\bar\Delta_n^{\epsilon,\beta+\alpha}$ such that $$(\a^*_\epsilon,\b^*_\epsilon,\c^*_\epsilon,\pmb\gamma^*_\epsilon)\in \arg\inf_{(\a,\b,\c,\pmb\gamma)\in\mathcal{A}_2\times\bar\Delta_n^{\epsilon,\beta+\alpha}}
R({\bf a}, {\bf b},\c,\pmb\gamma).$$
Note that $$\begin{aligned}\lim\limits_{\gamma_0\downarrow0}\mathrm{RVaR}_{\gamma_i,\gamma_0}(r_{a_i,b_i,c_i}(X_i))&=\lim\limits_{\gamma_0\downarrow0}\frac{1}{\gamma_0}\int^{\gamma_i+\gamma_0}_{\gamma_i}\VaR_{1-t}(r_{a_i,b_i,c_i}(X_i))\mathrm dt\\&=\VaR_{1-\gamma_i}(r_{a_i,b_i,c_i}(X_i)).   
\end{aligned}$$
We define $\mathrm{RVaR}_{\gamma_i,0}(r_{a_i,b_i,c_i}(X_i))=\VaR_{1-\gamma_i}(r_{a_i,b_i,c_i}(X_i)).$ Also, $\mathrm{RVaR}_{0,0}(r_{a_i,b_i,c_i}(X_i))=\infty$ for $a_i+c_i>0$ and $\operatorname{ess\,sup}X_i=\infty.$  With this situation, $\mathrm{RVaR}_{\gamma_i,\gamma_0}(r_{a_i,b_i,c_i}(X_i))$ is continuous with respect to $(\a,\b,\c,\pmb\gamma)$ over $\mathcal A_2\times\bar\Delta_n^{0,\beta+\alpha}.$ Then $R({\bf a}, {\bf b},\c,\pmb\gamma)$ is a continuous function of $(\a,\b,\c,\pmb\gamma)$ over $\mathcal A_2\times\bar\Delta_n^{0,\beta+\alpha}.$ According to the fact that $\mathcal A_2\times\bar\Delta_n^{0,\beta+\alpha}$ is a closed set,  there exists $$(\a^*,\b^*,\c^*)=\arg\inf_{(\a,\b,\c)\in\mathcal{A}_2}
  \left\{\inf_{\pmb{\gamma}\in(\beta+\alpha)\Delta_n,\gamma_0\geq \alpha}R({\bf a}, {\bf b},\c, \pmb\gamma)\right\}.$$
  We complete the proof. 
\end{proof}

\section{Proofs of Section \ref{sec:4}}
 \begin{proof}[Proof of Theorem \ref{thvar}] 
(i) In light of Lemma \ref{lem:1}, we have for $n=2$,
\begin{align*}&\inf_{\f\in\mathcal{I}^n}V(\mathbf f)\\
&=\inf_{\f\in\mathcal{I}^n}\inf_{\pmb{\gamma} \in(1-\alpha) \Delta_n} \sum_{i=1}^n \left(\mathrm{VaR}_{\alpha_i}\left(X_i\right)-f_i(\mathrm{VaR}_{\alpha_i}\left(X_i\right))+\RVaR_{\gamma_i, \gamma_0}\left(f_i(X_i)\right)\right)\\
&=\inf_{\pmb{\gamma} \in(1-\alpha) \Delta_n} \inf_{\f\in\mathcal{I}^n}\sum_{i=1}^n \left(\mathrm{VaR}_{\alpha_i}\left(X_i\right)-f_i(\mathrm{VaR}_{\alpha_i}\left(X_i\right))+\RVaR_{\gamma_i, \gamma_0}\left(f_i(X_i)\right)\right).
\end{align*}
Let $a_i=\VaR_{\alpha_i}(X_i)-f_i(\VaR_{\alpha_i}(X_i))$ and $b_i=\VaR_{\alpha_i}(X_i).$ Then it follows that $f_i(\mathrm{VaR}_{\alpha_i}\left(X_i\right))=g_{a_i,b_i}(\mathrm{VaR}_{\alpha_i}\left(X_i\right))$ and $f_i(X_i)\geq g_{a_i,b_i}(X_i),$ which implies
$$\RVaR_{\gamma_i, \gamma_0}\left(f_i(X_i)\right)-f_i(\mathrm{VaR}_{\alpha_i}\left(X_i\right))\geq \RVaR_{\gamma_i, \gamma_0}\left(g_{a_i,b_i}(X_i)\right)-g_{a_i,b_i}(\mathrm{VaR}_{\alpha_i}\left(X_i\right));$$  see Figure \ref{fig:VaR} (i).
Hence, we have 
$$\inf_{\f\in\mathcal{I}^n}V(\mathbf f)\geq 
  \inf_{\pmb{\gamma}\in(1-\alpha)\Delta_n} \inf_{(\a,\b)\in\mathcal A_1 }\overline G(\a, \b, {\pmb \gamma})=\inf_{(\a,\b)\in\mathcal A_1 }
  \inf_{\pmb{\gamma}\in(1-\alpha)\Delta_n} \overline G(\a, \b, {\pmb \gamma}).$$
  The inverse inequality is trivial, similar to the proof of Theorem \ref{th:RVaR} (i). We obtain the desired result.

  (ii) Let $f_i \in \mathcal{I}_{cv}^n$ and consider $X_i$ with $F_{X_i} \in \mathcal{M}_{cx}^{1-\beta-\alpha}$.  
To show that $F_{f_i(X_i)} \in \mathcal{M}_{cx}^{1-\beta-\alpha}$, we can follow the same method as in Theorem \ref{th:RVaR}, which establishes that applying a convex ceded loss function to a variable with a concave-tail distribution preserves the concavity in the tail.  By Lemma \ref{lem:1}, we have 
\begin{align*}&\inf_{\f\in\mathcal{I}_{cv}^n}V(\mathbf f)\\
&=\inf_{\f\in\mathcal{I}_{cv}^n}\inf_{\pmb{\gamma} \in(1-\alpha) \Delta_n} \sum_{i=1}^n \left(\mathrm{VaR}_{\alpha_i}\left(X_i\right)-f_i(\mathrm{VaR}_{\alpha_i}\left(X_i\right))+\RVaR_{\gamma_i, \gamma_0}\left(f_i(X_i)\right)\right)\\
&=\inf_{\pmb{\gamma} \in(1-\alpha) \Delta_n} \inf_{\f\in\mathcal{I}_{cv}^n}\sum_{i=1}^n \left(\mathrm{VaR}_{\alpha_i}\left(X_i\right)-f_i(\mathrm{VaR}_{\alpha_i}\left(X_i\right))+\RVaR_{\gamma_i, \gamma_0}\left(f_i(X_i)\right)\right).
\end{align*}
Let  $a_i=\frac{f_i(\VaR_{\alpha_i}(X_i))}{\VaR_{\alpha_i}(X_i)}$ and $b_i=\VaR_{\alpha_i}(X_i).$ Then we have$$\RVaR_{\gamma_i, \gamma_0}\left(f_i(X_i)\right)-f_i(\mathrm{VaR}_{\alpha_i}\left(X_i\right))\geq \RVaR_{\gamma_i, \gamma_0}\left(l_{a_i,b_i}(X_i)\right)-l_{a_i,b_i}(\mathrm{VaR}_{\alpha_i}\left(X_i\right));$$ see Figure \ref{fig:VaR} (ii).
The rest of the proof is the same as that of (i). Hence, it is omitted. 

(iii)  As already verified in Theorem \ref{th:RVaR}, one can  check that if $X_i \sim F_i \in \mathcal{M}_{cv}^{\alpha}$ and $f_i \in \mathcal{I}_{cx}$, then the cumulative distribution function of $f_i(X_i)$ is concave beyond its $\alpha$-quantile.
 Hence, by Lemma \ref{lem:1}, we have 
\begin{align*}&\inf_{\f\in\mathcal{I}_{cx}^n}V(\mathbf f)\\
&=\inf_{\f\in\mathcal{I}_{cx}^n}\inf_{\pmb{\gamma} \in(1-\beta) \Delta_n} \sum_{i=1}^n \left(\mathrm{VaR}_{\alpha_i}\left(X_i\right)-f_i(\mathrm{VaR}_{\alpha_i}\left(X_i\right))+\RVaR_{\gamma_i, \gamma_0}\left(f_i(X_i)\right)\right)\\
&=\inf_{\pmb{\gamma} \in(1-\beta) \Delta_n} \inf_{\f\in\mathcal{I}_{cx}^n}\sum_{i=1}^n \left(\mathrm{VaR}_{\alpha_i}\left(X_i\right)-f_i(\mathrm{VaR}_{\alpha_i}\left(X_i\right))+\RVaR_{\gamma_i, \gamma_0}\left(f_i(X_i)\right)\right).
\end{align*}
Let $a_i=(f_i)_+'(\VaR_{\alpha_i}(X_i))$ and $b_i=\VaR_{\alpha_i}(X_i)-\frac{f_i(\VaR_{\alpha_i}(X_i))}{a_i},$ where $(f_i)_+'(x)$ is the right derivative of $f_i$ at $x$.  Then it follows that $f_i(\mathrm{VaR}_{\alpha_i}\left(X_i\right))=h_{a_i,b_i}(\mathrm{VaR}_{\alpha_i}\left(X_i\right))$ and $f_i(X_i)\geq h_{a_i,b_i}(X_i),$ which implies
$$\RVaR_{\gamma_i, \gamma_0}\left(f_i(X_i)\right)-f_i(\mathrm{VaR}_{\alpha_i}\left(X_i\right))\geq \RVaR_{\gamma_i, \gamma_0}\left(h_{a_i,b_i}(X_i)\right)-h_{a_i,b_i}(\mathrm{VaR}_{\alpha_i}\left(X_i\right));$$
 see Figure \ref{fig:VaR}  (iii).
The rest of  the proof is exactly the same as that of (i). The details are omitted. 
\end{proof}
\begin{figure}[h]
\centering
 \includegraphics[width=16cm]{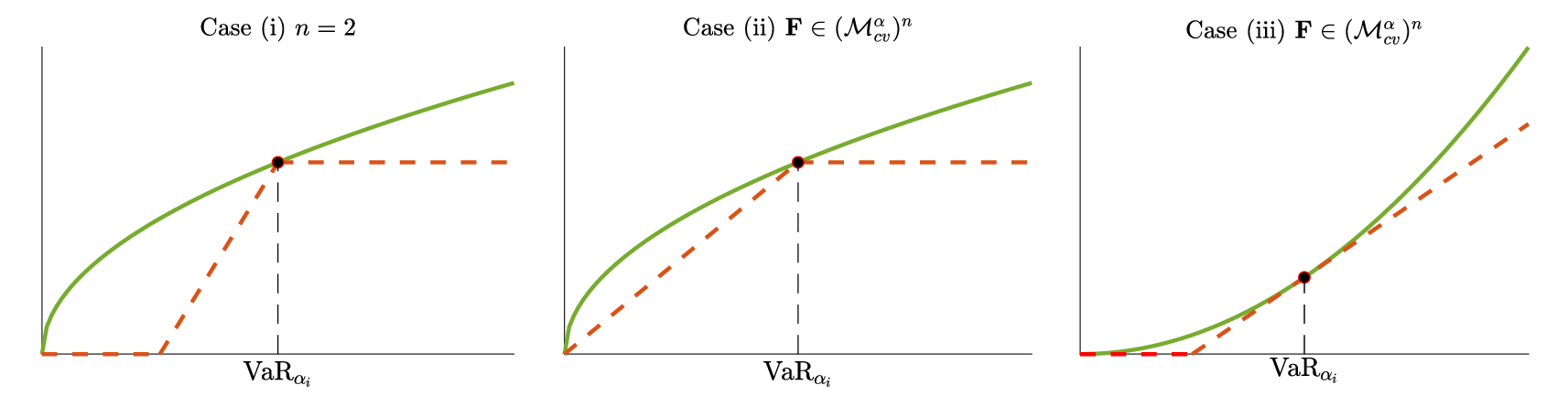}
 \captionsetup{font=small}
\caption{ \small 
(i) $g_{a_i,b_i}$ for the general case with $n=2$;  
(ii) $l_{a_i,b_i}$ for $\mathbf F \in (\mathcal M_{cx}^{\alpha})^n$;  
(iii) $h_{a_i,b_i}$ for $\mathbf F \in (\mathcal M_{cv}^{\alpha})^n$.
}\label{fig:VaR}
\end{figure}

 \begin{proof}[Proof of Proposition \ref{n2}]
By the result in \cite{M81} and the continuity of $F_1^{-1}$ and $F_2^{-1}$, we have
 $$\begin{aligned}
 V(\mathbf f)=&\VaR_{\alpha_1}(X_1)+\VaR_{\alpha_2}(X_2)- f(\VaR_{\alpha_1}(X_1))- f_2(\VaR_{\alpha_2}(X_2))\\&+\inf_{t\in[0,1-\alpha]}\left\{f_1(\VaR_{\alpha+t}(X_1))+ f_2(\VaR_{1-t}(X_2))\right\}.\end{aligned}$$
 Hence, \begin{align*}
 \inf_{(f_1,f_2)\in \mathcal{I}^2}V(\mathbf f)&=\VaR_{\alpha_1}(X_1)+\VaR_{\alpha_2}(X_2)+\inf_{t\in[0,1-\alpha]}\inf_{(f_1,f_2)\in \mathcal{I}^2}\left\{f_1(\VaR_{\alpha+t}(X_1))\right.\\&\left. + f_2(\VaR_{1-t}(X_2))-f_1(\VaR_{\alpha_1}(X_1))- f_2(\VaR_{\alpha_2}(X_2))\right\}.
 \end{align*}
Let $a_i=\VaR_{\alpha_i}(X_i)-f_i(\VaR_{\alpha_i}(X_i))$ and $b_i=\VaR_{\alpha_i}(X_i),~i=1,2.$
Using the above $a_i, b_i$, we have $f_i(\VaR_{\alpha_i}(X_i))=l_{a_i,b_i}(\VaR_{\alpha_i}(X_i))$ and $f_i(x)\geq g_{a_i,b_i}(x)$ for all $x\geq 0.$ Hence, we have  $$\begin{aligned}&f_1(\VaR_{\alpha+t}(X_1))-f_1(\VaR_{\alpha_1}(X_1))+f_2(\VaR_{1-t}(X_2))- f_2(\VaR_{\alpha_2}(X_2))\\\geq &g_{a_1,b_1}(\VaR_{\alpha+t}(X_1))- g_{a_1,b_1}(\VaR_{\alpha_1}(X_1))+g_{a_2,b_2}(\VaR_{1-t}(X_2))- g_{a_2,b_2}(\VaR_{\alpha_2}(X_2)),\end{aligned}$$
 which implies $$\inf_{(f_1,f_2)\in \mathcal{I}^2}V(\mathbf f)\geq \inf_{(a_1,a_2,b_1,b_2)\in \mathcal A}  \inf_{t\in[0,1-\alpha]} \overline G_1(a_1,a_2,b_1,b_2,t).$$
The inverse inequality holds trivially.  We complete the proof.
\end{proof}

 \section{Proofs of Section \ref{sec:5}}
 \begin{proof}[Proof of Proposition \ref{CLT}]
Since $\{X_i\}_{i=1}^n$ are independent random variables, and $f_i(X_i)$ depends only on $X_i$, the sequence $\{f_i(X_i)\}_{i=1}^n$ is also a sequence of independent random variables.
Since each $f_i$ is 1-Lipschitz continuous, and $\mathbb{E}[ X_i^{2}] < \infty$, we have $
   \mathbb{E}[f_i(X_i)^2] \leq \mathbb{E}[X_i^2] < \infty.$
   Hence,  $\mathrm{Var}(f_i(X_i)) = \mathbb{E}[f_i(X_i)^2] - (\mathbb{E}[f_i(X_i)])^2$ exists and is finite.
Define $$Z_n = \frac{S_n - \mathbb{E}[S_n]}{\sqrt{\mathrm{Var}(S_n)}}.$$ 
We aim to show that $Z_n \xrightarrow{d} \mathcal{N}(0, 1).$ Note that 
$\mathbb{E}[S_n] = \sum_{i=1}^n \mathbb{E}[f_i(X_i)],$ and 
$\mathrm{Var}(S_n) = \sum_{i=1}^n \mathrm{Var}(f_i(X_i)).$

To verify the Lindeberg condition \citep[see, e.g.,][Theorem~27.3]{Billingsley1995}, recall that for any $\epsilon>0$,
\[
\frac{1}{\mathrm{Var}(S_n)} \sum_{i=1}^n \mathbb{E}\Big[(f_i(X_i) - \mathbb{E}[f_i(X_i)])^2 \, \mathbb{I}_{\{|f_i(X_i) - \mathbb{E}[f_i(X_i)]| > \epsilon \sqrt{\mathrm{Var}(S_n)}\}}\Big] \to 0, \quad n\to\infty.
\]
Using the 1-Lipschitz property of $f_i$, we have
$
|f_i(X_i) - f_i(\mu)| \le |X_i - \mu|,
$
which implies
$$
|f_i(X_i) - \mathbb{E}[f_i(X_i)]| =|f_i(X_i) - f_i(\mu)+f_i(\mu)-\mathbb{E}[f_i(X_i)]|  \leq |X_i - \mu| + \mathbb{E}[|X_i - \mu|].$$
Define 
\[
A_i := \{|f_i(X_i) - \mathbb{E}[f_i(X_i)]| > \epsilon \sqrt{\mathrm{Var}(S_n)}\}, \quad 
B_i := \{|X_i - \mu| + \mathbb{E}[|X_i - \mu|] > \epsilon \sqrt{\mathrm{Var}(S_n)}\}.
\]
Then it is  clear that $A_i \subseteq B_i$. Note that
\[
\mathrm{Var}(S_n) = \sum_{i=1}^n \mathbb{E}[(f_i(X_i) - \mathbb{E}[f_i(X_i)])^2].
\]
Hence, we can bound the Lindeberg term as
\[
\begin{aligned}
& \frac{1}{\mathrm{Var}(S_n)} \sum_{i=1}^n \mathbb{E}\big[(f_i(X_i) - \mathbb{E}[f_i(X_i)])^2 \, \mathbb{I}_{A_i}\big] \\
&= \sum_{i=1}^n \frac{\mathbb{E}[(f_i(X_i) - \mathbb{E}[f_i(X_i)])^2]}{\mathrm{Var}(S_n)} \cdot 
\frac{\mathbb{E}[(f_i(X_i) - \mathbb{E}[f_i(X_i)])^2 \, \mathbb{I}_{A_i}]}{\mathbb{E}[(f_i(X_i) - \mathbb{E}[f_i(X_i)])^2]} \\
&\le \max_{i \in [n]} \frac{\mathbb{E}[(f_i(X_i) - \mathbb{E}[f_i(X_i)])^2 \, \mathbb{I}_{A_i}]}{\mathbb{E}[(f_i(X_i) - \mathbb{E}[f_i(X_i)])^2]}.
\end{aligned}
\]
It is easy to verify that
\[
\mathbb{E}[|X_i-\mu|] \le 2 \mathbb{E}[X_i] < \infty, \quad 
\mathbb{E}[|X_i-\mu|^2] \le 4 \mathbb{E}[X_i^2] < \infty.
\]
Hence,
\[
\mathbb{E}\big[(f_i(X_i) - \mathbb{E}[f_i(X_i)])^2\big] 
\le \mathbb{E}\big[(|X_i-\mu| + \mathbb{E}[|X_i-\mu|])^2\big] < \infty.
\]
If $f_i$ is a constant function, then $f_i(X_i)- \mathbb{E}[f_i(X_i)] \equiv 0$.  
The convention $0/0=0$  ensures   that this case does not cause any issues. Clearly,
\[
0 \le \frac{\mathbb{E}[(f_i(X_i) - \mathbb{E}[f_i(X_i)])^2 \, \mathbb{I}_{A_i}]}{\mathbb{E}[(f_i(X_i) - \mathbb{E}[f_i(X_i)])^2]} \le 1.
\]
Since $\mathrm{Var}(S_n) \to \infty$ as $n \to \infty$, 
  using Markov's inequality, we have
\[
\mathbb{P}\big(|X_i-\mu| + \mathbb{E}[|X_i-\mu|] > \epsilon \sqrt{\mathrm{Var}(S_n)}\big) 
\le \frac{2\mathbb{E}[|X_i-\mu|]}{\epsilon \sqrt{\mathrm{Var}(S_n)}} \to 0, \quad n \to \infty.
\]
Since $A_i \subseteq B_i$, it follows that $\mathbb{P}(A_i) \le \mathbb{P}(B_i) \to 0$. 
By the Monotone Convergence Theorem,
\[
\lim_{n\to\infty}\frac{\mathbb{E}[(f_i(X_i) - \mathbb{E}[f_i(X_i)])^2 \, \mathbb{I}_{A_i}]}{\mathbb{E}[(f_i(X_i) - \mathbb{E}[f_i(X_i)])^2]} 
= \frac{\mathbb{E}[\lim_{n\to\infty} (f_i(X_i) - \mathbb{E}[f_i(X_i)])^2 \, \mathbb{I}_{A_i}]}{\mathbb{E}[(f_i(X_i) - \mathbb{E}[f_i(X_i)])^2]} = 0.
\]
Thus, the Lindeberg condition holds for each $i\in[n]$. 
Since $\{f_i(X_i)\}_{i=1}^n$ are independent with finite variances, the Central Limit Theorem implies
\[
Z_n = \frac{S_n - \mathbb{E}[S_n]}{\sqrt{\mathrm{Var}(S_n)}} \xrightarrow{d} \mathcal{N}(0, 1).
\] We complete the proof.
\end{proof}
 \begin{proof}[Proof of Theorem \ref{RVaR_as}] 
Let   $M= \frac{1}{\alpha}\int_{\bar\alpha}^{\bar\beta} \Phi^{-1}(\gamma)\mathrm d\gamma.$ 
We prove the result in the following three steps:

\textbf{Step 1:}  For any $f_i \in \mathcal{I},$ let
$$
k_{f_i}(x) \triangleq \begin{cases}f_i(x), & 0 \leq x \leq \VaR_{\bar \alpha_i}(X), \\ \min \left\{x+f_i\left(\VaR_{\bar \alpha_i}(X)\right)-\VaR_{\bar \alpha_i}(X), b_i\right\}, & x>\VaR_{\bar \alpha_i}(X),\end{cases}
$$
where $b_i \geq f_i\left(\VaR_{\bar \alpha_i}(X)\right)$ is determined by $\ES_{\bar\alpha_i}(f_i(X))=\ES_{\bar\alpha_i} \left(k_{f_i}(X)\right).$ Note that
$$
f_i(x) \leq x-\VaR_{\bar \alpha_i}(X)+f_i\left(\VaR_\alpha(X)\right), \quad \forall x \geq \operatorname{VaR}_{\bar \alpha_i}(X).
$$
Then,
$$
\begin{aligned}
\ES_{\bar \alpha_i}(f_i(X)) & =\frac{1}{\alpha_i+\beta_i} \int_{\bar \alpha_i}^{1} \operatorname{VaR}_{s}(f_i(X)) \mathrm{d} s=\frac{1}{\alpha_i+\beta_i} \int_{\bar \alpha_i}^{1} f_i\left(\VaR_{s}(X)\right) \mathrm{d} s \\
&\leq \frac{1}{\alpha_i+\beta_i} \int_{\bar \alpha_i}^{1} \left(\operatorname{VaR}_{s}(X)-\operatorname{VaR}_{\bar \alpha_i}(X)+f\left(\VaR_{\bar \alpha_i}(X)\right)\right) \mathrm{d} s.
\end{aligned}
$$
Thus, there exists   $b_i\geq f_i\left(\VaR_\alpha(X)\right)$  such that
$$\begin{aligned}
\ES_{\bar\alpha_i}(f_i(X))&=  \frac{1}{\alpha_i+\beta_i} \int_{\bar\alpha_i}^{1}\min\{ \operatorname{VaR}_{s}(X)-\operatorname{VaR}_{\bar \alpha_i}(X)+f_i\left(\VaR_{\bar \alpha_i}(X)\right),b_i\} \mathrm{d} s\\& =\ES_{\bar\alpha_i} \left(k_{f_i}(X)\right).
\end{aligned}$$
In conjunction with 
$ \ES_{\bar\beta_i}(f_i(X)) \geq \ES_{\bar\beta_i}(k_{f_i}(X)),$ 
this yields
$
\mathrm{RVaR}_{\beta_i, \alpha_i}(f_i(X))
\leq
\mathrm{RVaR}_{\beta_i, \alpha_i}(k_{f_i}(X)).
$
 Next, we demonstrate $\mathbb{E}\left[k_{f_i}(X)\right]=\mathbb{E}[f_i(X)].$ Let  $U$ be uniformly distributed on $[0,1],$  then 
\begin{align*}
\mathbb{E}[f_i(X)] & =\mathbb{E}\left[\VaR_U(f_i(X))\right]=\mathbb{E}\left[f_i\left(\VaR_U(X)\right)\right] \\
& =\int_0^{\bar \alpha_i} f_i\left(\VaR_s(X)\right) \mathrm{d} s + \int_{\bar \alpha_i} ^{1 }f_i\left(\VaR_s(X)\right) \mathrm{d} s \\
& =\int_0^{\bar \alpha_i} k_{f_i}\left(\VaR_s(X)\right) \mathrm{d} s +(\alpha_i+\beta_i) \ES_{\bar \alpha_i} (f_i(X)) \\
& =\int_0^{\bar \alpha_i} k_{f_i}\left(\VaR_s(X)\right) \mathrm{d} s +(\alpha_i+\beta_i) \ES_{\bar \alpha_i} (k_{f_i}(X)) =\mathbb{E}\left[k_{f_i}(X)\right].
\end{align*}

Finally, by setting $t_0=b_i$ in Lemma \ref{lem:3}, let $G(x)=(x-\mathbb{E}\big[f_i(X_i)\big])^2,$ we have 
$\mathbb{P}(k_{f_i}(X)\leq t) \leq  \mathbb{P}( f_i(X)\leq t) ,$ for  $t< b_i$ and $\mathbb{P}(k_{f_i}(X)> t) \leq  \mathbb{P}( f_i(X)> t) $ for $ t \geq  b_i $, then $k_{f_i}(X) \leq_{c x} f_i(X).$  Then by  Lemma \ref{lem:3}, we know that  $\var( k_{f_i}(X))\leq  \var(f_i(X)).$

{\bf  Step 2:} Let  $0\leq c_i\leq f_i(\VaR_{\bar\alpha_i}(X))$ and $d_i= \VaR_{\bar\alpha_i}(X)- f_i(\VaR_{\bar\alpha_i}(X))+c_i.$ Define $$\widetilde k_{f_i}(x,c_i)= \left\{\begin{aligned} &k_{f_i}(x), &x>\VaR_{\bar \alpha_i}(X),\\ &x\wedge c_i+(x-d_i)_+,&x\leq \VaR_{\bar \alpha_i}(X). \end{aligned}\right.$$
It is straightforward to verify that $0=\widetilde k_{f_i}(x,0)\leq k_{f_i}(x)$ and $\widetilde k_{f_i}(x, f_i(\VaR_{\bar\alpha_i}(X)) )\geq k_{f_i}(x).$ Further, since $\widetilde k_{f_i}(x,c_i)$ is increasing and continuous with respect to $c_i,$ then there exists a $c_i^*\in[0, f_i(\VaR_{\bar \alpha_i}(X))]$ such that $\E[\widetilde k_{f_i}(X,c^*_i) ]= \E[k_{f_i}(X)].$  

Moreover, for $y< c_i^*,$ we have $\mathbb P(\widetilde k_{f_i}(X,c^*_i)\leq y) \leq \mathbb P(k_{f_i}(X)\leq y).$  Let $d_i^*= \VaR_{\bar\alpha_i}(X)-f_i(\VaR_{\bar\alpha_i}(X))+c^*_i,$ then  for $y\geq c^*_i,$  
$$\begin{aligned}
\mathbb{P}(\widetilde k_{f_i}(X,c^*_i)>y)= & \mathbb{P}(X \geq c^*_i, \widetilde k_{f_i}(X,c^*_i)>y) \\
= & \mathbb{P}(X>\VaR_{\bar\alpha_i}(X), k_{f_i}(X)>y) 
 +\mathbb{P}(c_i^* \leq X \leq \VaR_{\bar\alpha_i}(X), X-d^*_i+c^*_i>y) \\
\leq & \mathbb{P}(X>\VaR_{\bar\alpha_i}(X),k_ {f_i}(X)>y)  +\mathbb{P}\left(c_i^* \leq X \leq \VaR_{\bar\alpha_i}(X), k_{f_i}(X)>y\right) \\
= & \mathbb{P}(k_{f_i}(X)>y). 
\end{aligned}$$
Then by Lemma \ref{lem:3}, we know that $\var(\widetilde k_{f_i}(X,c^*_i))\leq  \var(k_{f_i}(X)).$ 

\textbf{Step 3:}  Define $$k^*_{f_i}(x,a_i)=(x-a_i)_+\wedge (\VaR_{\bar\alpha_i}(X)-f_i(\VaR_{\bar\alpha_i}(X))+b_i-a_i)$$ with $0\leq a_i\leq \VaR_{\bar\alpha_i}(X).$
Let $\theta_i\in[0,1],$ $\bar\theta_i=1-\theta_i$ and $\gamma_i\in[0,1],$ such that $$\VaR_{\theta_i}(X)=(\VaR_{\bar\alpha_i}(X)-f_i(\VaR_{\bar\alpha_i}(X))+b_i,$$ and $$\VaR_{\gamma_i}(X)=(\VaR_{\bar\alpha_i}(X)-f_i(\VaR_{\bar\alpha_i}(X)).$$
Then we have $$\begin{aligned}&\E[ k^*_{f_i}(\mathrm{RVaR}_{\beta_i, \alpha_i} (X),a_i)-k^*_{f_i}(X,a_i)]\\=~&  \mathrm{RVaR}_{\beta_i, \alpha_i} (X)-a_i- \E[(X-a_i)_+\wedge (\VaR_{\gamma_i}(X)+b_i-a_i)]\\=~&\E[(\mathrm{RVaR}_{\beta_i, \alpha_i} (X)-a_i)\wedge\max(\mathrm{RVaR}_{\beta_i, \alpha_i}(X)-X , \mathrm{RVaR}_{\beta_i, \alpha_i}(X)-\VaR_{\theta_i}(X))]\\:=~&g(a_i).\end{aligned}
$$ 
 It is clear that 
 $$\begin{aligned}
     g(\VaR_{\gamma_i}(X))=&\mathrm{RVaR}_{\beta_i, \alpha_i}(X)-\E[X \id_{\{\VaR_{\gamma_i}(X)\leq X\leq\VaR_{\theta_i}(X)\}}]-\gamma_i(\VaR_{\gamma_i}(X)-\bar\theta_i\VaR_{\theta_i}(X),
     \end{aligned}$$ and $$g(\VaR_{\bar \alpha_i}(X))=\mathrm{RVaR}_{\beta_i, \alpha_i}(X)-\bar\alpha_i\VaR_{\bar\alpha_i}(X)-\E[X \id_{\{\VaR_{\bar\alpha_i}(X)\leq X\leq\VaR_{\theta_i}(X)\}}]-\bar\theta_i\VaR_{\theta_i}(X).$$
   Also, we have 
 $$\begin{aligned}
     &\E[ \widetilde k_{f_i}(\mathrm{RVaR}_{\beta_i, \alpha_i} (X),c^*_i)-\widetilde k_{f_i}(X,c^*_i)]\\=~&\mathrm{RVaR}_{\beta_i, \alpha_i} (X)-\VaR_{\bar \alpha_i}(X)+f_i(\VaR_{\bar \alpha_i}(X))-\E[\widetilde k_{f_i}(X,c^*_i)]\\=~&\mathrm{RVaR}_{\beta_i, \alpha_i} (X)-\E[X+\VaR_{\bar \alpha_i}(X)-f_i(\VaR_{\bar \alpha_i}(X))\id_{\{X\in[0,c^*_i]\}}]\\&-\E[X\id _{\{X\in[d_i^*,\VaR_{\theta_i}(X)]\}}]-d_i^*(F_{X}(d_i^*)-F_{X}(c_i^*))-\bar\theta_i\VaR_{\theta_i}(X).
 \end{aligned}$$
Therefore, $$\begin{aligned}
     g(\VaR_{\bar\alpha_i}(X)-f_i(\VaR_{\bar\alpha_i}(X))\geq\E[ \widetilde k_{f_i}(\mathrm{RVaR}_{\beta_i, \alpha_i} (X),c^*_i)-\widetilde k_{f_i}(X,c^*_i)]\geq g(\VaR_{\bar \alpha_i}(X)).
 \end{aligned}$$
Together with the fact that  $g(a_i)$ is decreasing and continuous with respect to $ a_i,$ then there exist $ \VaR_{\bar\alpha_i}(X)-f_i(\VaR_{\bar\alpha_i}(X))\leq  a^*_i \leq \VaR_{\bar\alpha_i}(X)$  such that  $$\E[k^*_{f_i}(\mathrm{RVaR}_{\beta_i, \alpha_i} (X),a^*_i)-k^*_{f_i}(X,a^*_i)]=\E[ \widetilde k_{f_i}(\mathrm{RVaR}_{\beta_i, \alpha_i} (X),c^*_i)-\widetilde k_{f_i}(X,c^*_i)]. $$ 
For $y< \mathrm{RVaR}_{\beta_i, \alpha_i} (X)-a_i^*,$  we have 
$$\begin{aligned}
&\mathbb{P}(k^*_{f_i}(\mathrm{RVaR}_{\beta_i, \alpha_i} (X),a^*_i)-k^*_{f_i}(X,a^*_i)\leq y)\\=&~  \mathbb{P}(X \geq a_i^*,  k^*_{f_i}(\mathrm{RVaR}_{\beta_i, \alpha_i} (X),a^*_i)-k^*_{f_i}(X,a^*_i) \leq y) \\
= &~ \mathbb{P}(X>\VaR_{\bar \alpha_i}(X),  \widetilde k_{f_i}(\mathrm{RVaR}_{\beta_i, \alpha_i} (X),c^*_i)-\widetilde k_{f_i}(X,c^*_i)\leq y) 
 \\&+\mathbb{P}(a_i^*  \leq X\leq \VaR_{\bar \alpha_i}(X), \mathrm{RVaR}_{\beta_i, \alpha_i} (X)-X\leq y) \\
\leq &~  \mathbb{P}(X>\VaR_{\bar \alpha_i}(X),  \widetilde k_{f_i}(\mathrm{RVaR}_{\beta_i, \alpha_i} (X),c^*_i)-\widetilde k_{f_i}(X,c^*_i)\leq y) 
 \\&+\mathbb{P}(a_i^*  \leq X\leq \VaR_{\bar \alpha_i}(X),  \widetilde k_{f_i}(\mathrm{RVaR}_{\beta_i, \alpha_i} (X),c^*_i)-\widetilde k_{f_i}(X,c^*_i) \leq y)  \\
 =&~  \mathbb{P}( \widetilde k_{f_i}(\mathrm{RVaR}_{\beta_i, \alpha_i} (X),c^*_i)-\widetilde k_{f_i}(X,c^*_i)\leq y). 
\end{aligned}$$
On the other hand, for $y\geq \mathrm{RVaR}_{\beta_i, \alpha_i} (X) -a_i^*,$ we have $$0=\mathbb{P}(k^*_{f_i}(\VaR_{\bar \alpha_i}(X),a^*_i)-k^*_{f_i}(X,a^*_i)> y) \leq \mathbb{P}(\widetilde k_{f_i}(\VaR_{\bar\alpha_i}(X),c^*_i)-\widetilde k_{f_i}(X,c^*_i)> y). $$
Thus, we have $$\var\left(k^*_{f_i}(\mathrm{RVaR}_{\beta_i, \alpha_i} (X),a^*_i)-k^*_{f_i}(X,a^*_i)\right)\leq  \var(\widetilde k_{f_i}\left(\mathrm{RVaR}_{\beta_i, \alpha_i} (X),c^*_i)-\widetilde k_{f_i}(X,c^*_i)\right).$$
To summarize,  
 \begin{align*}& \sum_{i=1}^n   \left( \mathrm{RVaR}_{\beta_i, \alpha_i} (X_i)-    \mathrm{RVaR}_{\beta_i, \alpha_i} (f_i(X_i))+\mathbb E[f_i(X_i)] \right)+  M\left(\sum_{i=1}^n \var(f_i(X_i))\right)^{1/2}  \\\geq & \sum_{i=1}^n   \left( \mathrm{RVaR}_{\beta_i, \alpha_i} (X_i)- \mathrm{RVaR}_{\beta_i, \alpha_i}  (k_{f_i}(X_i))+ \mathbb E[k_{f_i}(X_i)]  \right)+  M \left(\sum_{i=1}^n \var(k_{f_i}(X_i))\right)^{1/2}  \\ \geq &  \sum_{i=1}^n   \left( \mathrm{RVaR}_{\beta_i, \alpha_i} (X_i)- \mathrm{RVaR}_{\beta_i, \alpha_i}  (\widetilde k_{f_i}(X_i,c_i^*))+ \mathbb E[\widetilde k_{f_i}(X_i,c_i^*)]  \right)+  M  \left(\sum_{i=1}^n \var(\widetilde k_{f_i}(X_i,c^*_i))\right)^{1/2}\\=&\sum_{i=1}^n   \left( \mathrm{RVaR}_{\beta_i, \alpha_i} (X_i)- \big(\mathrm{RVaR}_{\beta_i, \alpha_i}  (\widetilde k_{f_i}(X_i,c_i^*))- \mathbb E[\widetilde k_{f_i}(X_i,c_i^*)]\big)  \right)+  M  \left(\sum_{i=1}^n \var(\widetilde k_{f_i}(X_i,c^*_i))\right)^{1/2}\\\geq & \sum_{i=1}^n   \left(  \mathrm{RVaR}_{\beta_i, \alpha_i} (X_i)-\big(\RVaR_{\beta_i, \alpha_i}(k^*_{f_i}(X_i,a^*_i))- \mathbb E[k^*_{f_i}(X_i,a_i^*)]) \right)+   M \left(\sum_{i=1}^n \var( k^*_{f_i}(X_i,a^*_i))\right)^{1/2}\\= & \sum_{i=1}^n   \left(  \mathrm{RVaR}_{\beta_i, \alpha_i} (X_i)-\RVaR_{\beta_i, \alpha_i}(k^*_{f_i}(X_i,a^*_i))+ \mathbb E[k^*_{f_i}(X_i,a_i^*)] \right)+   M \left(\sum_{i=1}^n\var( k^*_{f_i}(X_i,a^*_i))\right)^{1/2},\end{align*}
which  completes the proof. 
\end{proof}

\begin{proof}[Proof of Proposition \ref{prop:5}]For any $\mathbf f=(f_1,\dots,f_n)\in\mathcal I^n$,
it is clear that $$ \begin{aligned}&\sum_{i=1}^n \mathrm{VaR}_{\alpha_i}\left(T_{f_i,\pi_i}(X_i)\right)+ \mathrm{VaR}_\alpha\left(R(\mathbf f,\pi)\right)\\&= \sum_{i=1}^n \mathrm{VaR}_{\alpha_i}(X_i)-\sum_{i=1}^n f_{i}\left(\mathrm{VaR}_{\alpha_i}(X_i)\right)+ \mathrm{VaR}_\alpha\left(\sum_{i=1}^nf_i( X_i)\right).\end{aligned}$$
Let $a_i=\VaR_{\alpha_i}(X_i)-f_i(\VaR_{\alpha_i}(X_i))$ and $b_i=\VaR_{\alpha_i}(X_i).$ Then it follows that $f_i(\mathrm{VaR}_{\alpha_i}\left(X_i\right))=g_{a_i,b_i}(\mathrm{VaR}_{\alpha_i}\left(X_i\right))$ and $f_i(X_i)\geq g_{a_i,b_i}(X_i),$ which implies
$$ -\sum_{i=1}^n g_{a_i,b_i}\left(\mathrm{VaR}_{\alpha_i}(X_i)\right)+ \mathrm{VaR}_\alpha\left(\sum_{i=1}^ng_{a_i,b_i}( X_i)\right) \leq -\sum_{i=1}^n f_{i}\left(\mathrm{VaR}_{\alpha_i}(X_i)\right)+ \mathrm{VaR}_\alpha\left(\sum_{i=1}^nf_i( X_i)\right).$$ Then we get the desired result.
\end{proof} 

\begin{proof}[Proof of Theorem \ref{thme}]
From Proposition \ref{prop:5}, we immediately obtain $b^*_i = \VaR_{\alpha_i}(X_i)$.  Next, we  establish the monotonicity properties of the auxiliary functions. Direct computation yields:
\[
w'_i(a_i) = -S_{X_i}(a_i) \leq 0, \quad \text{and} \quad v'_i(a_i) = -2w_i(a_i) \leq 0,
\]
where $S_{X_i}(x) = 1 - F_{X_i}(x)$ is the survival function.
For any indemnity function $f_i = g_{a_i,b_i}$, the objective function becomes:
\[
\begin{aligned}
F(\a) &:= \sum_{i=1}^n \mathrm{VaR}_{\alpha_i}(X_i) - \sum_{i=1}^n f_i(\VaR_{\alpha_i}(X_i)) + \sum_{i=1}^n \mathbb{E}[f_i(X_i)] + \Phi^{-1}(\alpha) \left(\sum_{i=1}^n \mathrm{var}(f_i(X_i))\right)^{1/2} \\
&= \sum_{i=1}^n (a_i + w_i(a_i)) + \Phi^{-1}(\alpha) \left( \sum_{i=1}^n (v_i(a_i) - w_i(a_i)^2) \right)^{1/2}.
\end{aligned}
\]
Applying the first-order condition, we compute:
\[
\begin{aligned}
\frac{\partial F(\a)}{\partial a_i} &= F_{X_i}(a_i) \left(1 - \Phi^{-1}(\alpha) \left(\sum_{j=1}^n \left(v_j(a_j) - w_j(a_j)^2 \right) \right)^{-1/2} w_i(a_i) \right) \\
&= F_{X_i}(a_i) \left(1 - \Phi^{-1}(\alpha) G_i(\a)^{1/2} \right),
\end{aligned}
\]
where 
\[
G_i(\a) = \frac{w_i(a_i)^2}{\sum_{j=1}^n \left(v_j(a_j) - w_j(a_j)^2 \right)}.
\]
This implies \[
\frac{\partial F(\a)}{\partial a_i} \geq 0 \quad \Longleftrightarrow \quad 1 - \Phi^{-1}(\alpha) G_i(\a)^{1/2} \geq 0.
\]
To establish the monotonicity of $G_i(\a)$, consider the function
\[
\begin{aligned}
g_i(a_i) &:= w_i(a_i)^2 - S_{X_i}(a_i)v_i(a_i) \\
&= \left(\int_{a_i}^{\VaR_{\alpha_i}(X_i)} S_{X_i}(x) \mathrm{d} x\right)^2 - 2S_{X_i}(a_i) \int_{a_i}^{\VaR_{\alpha_i}(X_i)} (x - a_i) S_{X_i}(x) \mathrm{d} x.
\end{aligned}
\]
We observe that $g_i(\VaR_{\alpha_i}(X_i)) = 0$ and
\[
\frac{\partial g_i(a_i)}{\partial a_i} = \frac{\partial F_{X_i}(a_i)}{\partial a_i} v_i(a_i) \geq 0,
\]
which implies $g_i(a_i) \leq 0$ for $0 \leq a_i \leq \VaR_{\alpha_i}(X_i)$.

Now, differentiating $G_i(\a)$ with respect to $a_i$:
\[
\begin{aligned}
\frac{\partial G_i(\a)}{\partial a_i} &= \frac{-2w_i(a_i)S_{X_i}(a_i)\left(\sum_{j=1}^n (v_j(a_j) - w_j(a_j)^2)\right) + 2w_i(a_i)^3(1 - S_{X_i}(a_i))}{\left(\sum_{j=1}^n (v_j(a_j) - w_j(a_j)^2)\right)^2} \\
&= \frac{-2w_i(a_i)S_{X_i}(a_i)\left(\sum_{j\neq i} (v_j(a_j) - w_j(a_j)^2)\right) + 2w_i(a_i)^3 - 2w_i(a_i)S_{X_i}(a_i)v_i(a_i)}{\left(\sum_{j=1}^n (v_j(a_j) - w_j(a_j)^2)\right)^2} \leq 0.
\end{aligned}
\]
Hence, $G_i(\a)$ is decreasing in $a_i$, and consequently $\frac{\partial F(\a)}{\partial a_i}\geq0$. Therefore, the minimum of $F(\a)$ is attained at 
\[
\mathbf{a}^* = (a_1^*, \dots, a_n^*)
\]
with 
\[
a_i^* = \inf \left\{0 \leq a_i \leq \VaR_{\alpha_i}(X_i): 1 - \Phi^{-1}(\alpha) \cdot \frac{w_i(a_i)^2}{\sum_{j=1}^n \left(v_j(a_j) - w_j(a_j)^2 \right)} \geq 0 \right\},
\]
which completes the proof.
\end{proof}

\section{Proofs of Section \ref{sec:6}}

\begin{proof}[Proof of Lemma \ref{lem:4}]
The proofs of (ii) and (iii) are similar to those of (i), so we focus on (i).

\textbf{Step 1.}  
We first show that
\[
    \overline G_1({\bf u}, {\bf v}, t) \le \overline G_1({\bf a}, {\bf b}, t), \quad \forall t \in [0,1-\alpha],
\]
where $v_i = \VaR_{\alpha_i}(X_i)$ is fixed, and $u_i = v_i - g_{a_i,b_i}(v_i)$ for $i=1,2$. By Theorem 3.1 in \cite{CT13}, we have 
\[
    g_{u_i,v_i}(x) \le g_{a_i,b_i}(x), \quad \forall x \ge 0, \quad \text{and} \quad g_{u_i,v_i}(\VaR_{\alpha_i}(X_i)) = g_{a_i,b_i}(\VaR_{\alpha_i}(X_i)),~i=1,2.
\]  
Let $\theta_1=\alpha+t$ and $\theta_2=1-t.$ Hence,
\[
    \VaR_{\theta_i}(g_{u_i,v_i}(X_i)) \le \VaR_{\theta_i}(g_{a_i,b_i}(X_i)),
\]
and therefore
\begin{align*}
    \overline G_1(\mathbf{ u}, \mathbf{ v}, t) 
    &= \sum_{i=1}^{2} \big\{\VaR_{\alpha_i}(X_i) - \VaR_{\alpha_i}(g_{u_i,v_i}(X_i)) + \VaR_{\theta_i}(g_{u_i,v_i}(X_i)) \big\} \\
    &\le \sum_{i=1}^{2} \big\{\VaR_{\alpha_i}(X_i) - \VaR_{\alpha_i}(g_{a_i,b_i}(X_i)) + \VaR_{\theta_i}(g_{a_i,b_i}(X_i)) \big\} \\
    &= \overline G_1(\a, \b, t).
\end{align*}

\textbf{Step 2.}  
We now prove
\[
    \inf_{(\a,\b)\in \mathcal A_1} \inf_{t\in[0,1-\alpha]} \overline G_1(\a, \b, t)
    = \inf_{\bf u \in \mathcal A_1(\bf v)} \inf_{t\in[0,1-\alpha]} \overline G_1(\mathbf{ u}, \mathbf{ v}, t),
\]
based on Step 1. 

Let
\[
    S_{ab}^* = \arg\inf_{(\a,\b)\in \mathcal A_1} \inf_{t\in[0,1-\alpha]} \overline G_1(\a, \b, t),
    \quad
    S_{uv}^* = \arg\inf_{\bf u \in \mathcal A_1(\bf u)} \inf_{t\in[0,1-\alpha]} \overline G_1(\mathbf{ u}, \mathbf{ v},t).
\]
Assume $(u_1^*, u_2^*, t^*) \in S_{uv}^*$ but $(u_1^*, u_2^*, v_1, v_2, t^*) \notin S_{ab}^*$.  
Then there exists $(\hat{\a}, \hat{\b}, \hat{t}) \in S_{ab}^*$ such that
\[
    \overline G_1(\hat{\a}, \hat{\b}, \hat{t}) < \overline G_1(\mathbf{ u^*}, \mathbf{ v}, t^*).
\]
If $\hat{t} = t^*$, this is a contradiction. Otherwise, by Step 1, there exists $(\hat{\bf u}, \hat{t}) \in S_{uv}^*$ such that
\[
    \overline G_1(\hat{\mathbf{ u}}, \mathbf{ v}, \hat{t}) \le \overline G_1(\hat{\a}, \hat{\b}, \hat{t}) < \overline G_1(\mathbf{ u^*}, \mathbf{ v}, t^*),
\]
which is also a contradiction. Hence, $(u_1^*, u_2^*, t^*) \in S_{uv}^*$ implies $(u_1^*, u_2^*, v_1, v_2, t^*) \in S_{ab}^*$.

Conversely, if $(a_1^*, a_2^*, v_1, v_2, t^*) \in S_{ab}^*$ but $(a_1^*, a_2^*, t^*) \notin S_{uv}^*$, then there exists $(\widetilde{\mathbf{ u}}, t^*) \in S_{uv}^*$ such that
\[
    \overline G_1(\widetilde{\mathbf{ u}}, \mathbf{ v}, t^*) < \overline G_1(\a^*, \mathbf{ v}, t^*),
\]
again a contradiction.

Finally, by Step 1, for any $(a_1^*, a_2^*, b_1^*, b_2^*, t^*) \in S_{ab}^*$ with $b_1^* \neq v_1$ or $b_2^* \neq v_2$, there exists $(u_1', u_2', v_1, v_2, t^*) \in S_{ab}^*$ such that
\[
    \overline G_1(\a^*, \b^*, t^*) = \overline G_1(\mathbf{ u'}, \mathbf{ v}, t^*),
\]
and hence $(u_1', u_2', t^*) \in S_{uv}^*$.  

Therefore,
\[
    \inf_{(\a,\b)\in\mathcal A_1} \inf_{t\in[0,1-\alpha]} \overline G_1(\a, \b, t)
    = \inf_{\bf u \in \mathcal A_1(\bf v)} \inf_{t\in[0,1-\alpha]} \overline G_1(\mathbf{ u}, \mathbf{ v}, t),
\]
which yields the desired result.
\end{proof}

\begin{proof}[Proof of Proposition \ref{a12}]
  Let $\theta_1 = \alpha + t$ and $\theta_2 = 1 - t$. Consider the case where there exists $i \in \{1,2\}$ such that $\alpha \ge \alpha_i$, and assume without loss of generality that $\alpha \ge \alpha_2$.  
If $\alpha \ge \alpha_1$, then
\begin{align*}
    \VaR_{\theta_i}\big(g_{a_i,b_i}(X_i)\big) - \VaR_{\alpha_i}\big(g_{a_i,b_i}(X_i)\big) \ge 0, \quad i = 1,2.
\end{align*}
In this case, $a_i \in [0,b_i^*]$ for $i=1,2$, and $t^*=0$.  

If $\alpha < \alpha_1$ and $\alpha + t \le \alpha_1$, then $a_2$ still belongs to $[0, b_2^*]$. By the 1-Lipschitz property of $g_{a_1,b_1}$,  
\begin{align*}
    \VaR_{\alpha+t}\big(g_{a_1,b_1}(X_1)\big) - \VaR_{\alpha_1}\big(g_{a_1,b_1}(X_1)\big)
    &\ge \VaR_{\alpha+t}(X_1) - \VaR_{\alpha_1}(X_1) \\
    &\ge \VaR_{\alpha}(X_1) - \VaR_{\alpha_1}(X_1),
\end{align*}
and the infimum is achieved at $t^* = 0$ with $a_1 \in [0, \VaR_{\alpha}(X_1)]$. By symmetry, one can also obtain $t^* = 1 - \alpha$ for the case $\alpha \ge \alpha_1$.

Next, consider the case where $\alpha_i > \alpha$ for $i = 1,2$, and examine the corresponding admissible ranges of $t$ and the associated minimum values.

 {\bf Case (1): $t \in [0, \alpha_1 - \alpha)$.} 
In this case, we have $\alpha + t < \alpha_1$ and $1 - t > \alpha_2$. Since $1 - t > \alpha_2$,  
\[
    \VaR_{1-t}\big(g_{a_2,b_2}(X_2)\big)
    - \VaR_{\alpha_2}\big(g_{a_2,b_2}(X_2)\big) \ge 0,
\]
and the infimum is attained for $a_2 \in [0, \VaR_{\alpha_2}(X_2)]$.  
Because $\alpha + t < \alpha_1$, by the 1-Lipschitz condition,    
\begin{align*}
    \VaR_{\alpha+t}\big(g_{a_1,b_1}(X_1)\big)
    - \VaR_{\alpha_1}\big(g_{a_1,b_1}(X_1)\big)
    &\ge \VaR_{\alpha+t}(X_1) - \VaR_{\alpha_1}(X_1) \\
    &\ge \VaR_{\alpha}(X_1) - \VaR_{\alpha_1}(X_1),
\end{align*}
and the minimum is achieved at $t^* = 0$ and $a_1 \in [0, \VaR_{\alpha}(X_1)]$.  
Thus, the minimum value in this case is  
$
    \VaR_{\alpha}(X_1) + \VaR_{\alpha_2}(X_2).
$

 {\bf Case (2): $t \in [\alpha_1 - \alpha,\, 1 - \alpha_2]$.}  
In this case, $\alpha + t \ge \alpha_1$ and $1 - t \ge \alpha_2$. Recall that $\theta_1 = \alpha + t$ and $\theta_2 = 1 - t$. Then  
\[
    \VaR_{\theta_i}\big(g_{a_i,b_i}(X_i)\big)
    - \VaR_{\alpha_i}\big(g_{a_i,b_i}(X_i)\big) \ge 0, \quad i = 1,2,
\]
and the minimum value is  
$
    \VaR_{\alpha_1}(X_1) + \VaR_{\alpha_2}(X_2).
$
which is greater than that in Case 1.
    
    {\bf Case (3): $t\in(1-\alpha_2,1]$.} 
    In this case, $\alpha+t>\alpha_1,~1-t<\alpha_2.$ Since $\alpha+t>\alpha_1,$ 
    \begin{align*}
        \VaR_{\alpha+t}(g_{a_1,b_1}( X_1))-\VaR_{\alpha_1}(g_{a_1,b_1}(X_1))\geq0,
    \end{align*}
    which is obtained for $a_1\in[0,\VaR_{\alpha_1}(X_1)].$ Since $1-t<\alpha_2,$ by the 1-Lipschitz, 
    \begin{align*}
        \VaR_{1-t}(g_{a_2,b_2}( X_2))-\VaR_{\alpha_2}(g_{a_2,b_2}( X_2))&\geq \VaR_{1-t}( X_2)-\VaR_{\alpha_2}(X_2)\\&\geq \VaR_{\alpha}( X_2)-\VaR_{\alpha_2}(X_2),
    \end{align*}
    which is obtained for $t^*=1-\alpha$ and $a_2\in[0,\VaR_{\alpha}(X_1)].$ The minimum value is $\VaR_{\alpha_1}(X_1)+\VaR_{\alpha}(X_2).$\\
  To sum up, the minimum value is either $\VaR_{\alpha}(X_1)+\VaR_{\alpha_2}(X_2)$ or $\VaR_{\alpha_1}(X_1)+\VaR_{\alpha}(X_2)$, attained at $t=0$ or $t=1-\alpha$, respectively.   Consequently, the problem reduces to the case of one reinsurer and one insurer.
 The proof is complete.
\end{proof}

\begin{proof}[Proof of Proposition \ref{aa1a2}]
If $\alpha \ge \alpha_1 + \alpha_2 - 1$, the conclusion follows directly from Proposition~\ref{a12}. We therefore assume $\alpha < \alpha_1 + \alpha_2 - 1$ in the remainder of the proof.

We analyze the behavior of the objective function over different intervals of $t$:

\textbf{Case 1: $t \in [0, 1 - \alpha_2]$.} 
In this regime, we have $\alpha + t < \alpha_1$ and $1 - t \geq \alpha_2$. 
Following an argument analogous to Case (1) in the proof of Proposition~\ref{a12}, 
the minimum is attained at $t^* = 0$, with optimal value $\VaR_{\alpha}(X_1) + \VaR_{\alpha_2}(X_2)$.

\textbf{Case 2: $t \in [\alpha_1 - \alpha, 1 - \alpha]$.} 
In this parameter regime, the conditions $\alpha + t \geq \alpha_1$ and $1 - t < \alpha_2$ are satisfied. 
Following an argument parallel to Case (3) in the proof of Proposition~\ref{a12}, 
we find that the minimum is attained at the right endpoint $t^* = 1 - \alpha$. 
This yields the optimal value $\VaR_{\alpha_1}(X_1) + \VaR_{\alpha}(X_2)$.
  
  \textbf{Case (3): $t \in (1 - \alpha_2, \alpha_1 - \alpha)$.} 
In this case, $\alpha + t < \alpha_1$ and $1 - t < \alpha_2$. For any $t \in (1 - \alpha_2, \alpha_1 - \alpha)$, 
\[
\VaR_{1-t}(g_{a_2,b_2}(X_2)) - \VaR_{\alpha_2}(g_{a_2,b_2}(X_2)) \geq \VaR_{1-t}(X_2) - \VaR_{\alpha_2}(X_2),
\]
and
\[
\VaR_{\alpha+t}(g_{a_1,b_1}(X_1)) - \VaR_{\alpha_1}(g_{a_1,b_1}(X_1)) \geq \VaR_{\alpha+t}(X_1) - \VaR_{\alpha_1}(X_1),
\]
with the bounds attained by $a_1 \in [0, \VaR_{\alpha+t}(X_1)]$ and $a_2 \in [0, \VaR_{1-t}(X_2)]$. 
Hence, the minimum value for fixed $t \in (1 - \alpha_2, \alpha_1 - \alpha)$ is $\VaR_{\alpha+t}(X_1) + \VaR_{1-t}(X_2)$. 

The difference between this value and the minimum in Case (1) is
\begin{align*}
& \VaR_{\alpha+t}(X_1) - \VaR_{\alpha}(X_1) + \VaR_{1-t}(X_2) - \VaR_{\alpha_2}(X_2) \\
= & \VaR_{\alpha+t}(X_1) - \VaR_{\alpha+1-\alpha_2}(X_1) + \VaR_{\alpha+1-\alpha_2}(X_1) - \VaR_{\alpha}(X_1) \\
& + \VaR_{1-t}(X_2) - \VaR_{\alpha_2}(X_2).
\end{align*}
If $\VaR_{\alpha+t}(X_1) - \VaR_{\alpha+1-\alpha_2}(X_1) \geq \VaR_{\alpha_2}(X_2) - \VaR_{1-t}(X_2)$, 
then since $\VaR_{\alpha+1-\alpha_2}(X_1) - \VaR_{\alpha}(X_1) > 0$, the above expression is positive.  Otherwise, consider the difference between this value and the minimum in Case (2):
\begin{align*}
& \VaR_{\alpha+t}(X_1) - \VaR_{\alpha_1}(X_1) + \VaR_{1-t}(X_2) - \VaR_{\alpha}(X_2) \\
= & \VaR_{\alpha+t}(X_1) - \VaR_{\alpha_1}(X_1) + \VaR_{1-t}(X_2) - \VaR_{\alpha+1-\alpha_1}(X_2) \\
& + \VaR_{\alpha+1-\alpha_1}(X_2) - \VaR_{\alpha}(X_2).
\end{align*}
Since $\mathbf{F} \in \left(\mathcal{M}_{c x}^\alpha\right)^2$, we have
\[
\frac{1}{\alpha_2 + t - 1} \left( \VaR_{\alpha+t}(X_1) - \VaR_{\alpha+1-\alpha_2}(X_1) \right)
\geq \frac{1}{\alpha_1 - \alpha - t} \left( \VaR_{\alpha_1}(X_1) - \VaR_{\alpha+t}(X_1) \right),
\]
and
\[
\frac{1}{\alpha_1 - \alpha - t} \left( \VaR_{1-t}(X_2) - \VaR_{\alpha+1-\alpha_1}(X_2) \right)
\geq \frac{1}{\alpha_2 + t - 1} \left( \VaR_{\alpha_2}(X_2) - \VaR_{1-t}(X_2) \right).
\]
From the earlier assumption that $\VaR_{\alpha+t}(X_1) - \VaR_{\alpha+1-\alpha_2}(X_1) < \VaR_{\alpha_2}(X_2) - \VaR_{1-t}(X_2)$, 
it follows that $\VaR_{1-t}(X_2) - \VaR_{\alpha+1-\alpha_1}(X_2) > \VaR_{\alpha_1}(X_1) - \VaR_{\alpha+t}(X_1)$. 
Therefore, the above expression is positive. 
Thus, the minimum value in Case (3) is greater than the minima in both Cases (1) and (2). 
In summary, the global minimum is either $\VaR_{\alpha}(X_1) + \VaR_{\alpha_2}(X_2)$ or $\VaR_{\alpha_1}(X_1) + \VaR_{\alpha}(X_2)$, 
attained at $t = 0$ or $t = 1 - \alpha$, respectively. Consequently, the problem reduces to the case of one reinsurer and one insurer.
\end{proof}

\end{document}